\documentclass[runningheads, 11pt]{llncs}
\usepackage{graphicx}

\usepackage{amsmath,inputenc}
\usepackage{csvsimple}
\usepackage{rotating}
\usepackage[letterpaper, portrait, margin=1in]{geometry}

\usepackage[nocomma]{optidef}
\usepackage{listings}
\usepackage{comment}
\usepackage{appendix}
\usepackage [english]{babel}
\usepackage [autostyle, english = american]{csquotes}
\MakeOuterQuote{"}
\usepackage{color} 
\usepackage{amsmath}
\usepackage{thmtools, thm-restate}
\usepackage{hyperref}
\usepackage[ruled,vlined]{algorithm2e}

\DeclareMathOperator*{\argmax}{arg\,max}

\usepackage{tabularx}
    \newcolumntype{L}{>{\raggedright\arraybackslash}X}
\usepackage{array}
\definecolor{mygreen}{RGB}{28,172,0} 
\definecolor{mylilas}{RGB}{170,55,241}
\DeclareFixedFont{\ttb}{T1}{txtt}{bx}{n}{12} 
\DeclareFixedFont{\ttm}{T1}{txtt}{m}{n}{12}  
\usepackage{amssymb}
\usepackage[T1]{fontenc}
\usepackage{lipsum}

\newenvironment{hproof}{%
  \proof}{\endproof}

\usepackage{sectsty}
\sectionfont{\Large}
\subsectionfont{\large}
\subsubsectionfont{\large}

\newcommand{\norm}[1]{\left\lVert#1\right\rVert}
\usepackage{color}
\definecolor{deepblue}{rgb}{0,0,0.5}
\definecolor{deepred}{rgb}{0.6,0,0}
\definecolor{deepgreen}{rgb}{0,0.5,0}

\usepackage{listings}
\begin{document}
\title{Markets for Efficient Public Good Allocation with Social Distancing}
%
%
\author{Devansh Jalota\inst{1} \and Qi Qi\inst{2} \and
Marco Pavone\inst{1} \and
Yinyu Ye\inst{1}}
%
%
\institute{Stanford University, Stanford CA 94305, USA \\ \email{\{djalota,pavone,yyye\}@stanford.edu} \and The Hong Kong University of Science and Technology, Hong Kong \\ \email{kaylaqi@ust.hk} \\
}
\maketitle              
\begin{abstract}
\normalsize
Public goods are often either over-consumed in the absence of regulatory mechanisms, or remain completely unused, as in the Covid-19 pandemic, where social distance constraints are enforced to limit the number of people who can share public spaces. In this work, we plug this gap through market based mechanisms designed to efficiently allocate capacity constrained public goods. To design these mechanisms, we leverage the theory of Fisher markets, wherein each agent in the economy is endowed with an artificial currency budget that they can spend to avail public goods. While Fisher markets provide a strong methodological backbone to model resource allocation problems, their applicability is limited to settings involving two types of constraints - budgets of individual buyers and capacities of goods. Thus, we introduce a modified Fisher market, where each individual may have additional physical constraints, characterize its solution properties and establish the existence of a market equilibrium. Furthermore, to account for additional constraints we introduce a social convex optimization problem where we perturb the budgets of agents such that the KKT conditions of the perturbed social problem establishes equilibrium prices. Finally, to compute the budget perturbations we present a fixed point scheme and illustrate convergence guarantees through numerical experiments. Thus, our mechanism, both theoretically and computationally, overcomes a fundamental limitation of classical Fisher markets, which only consider capacity and budget constraints.

\end{abstract}

\section{Introduction} \label{Introduction}
A public good is a product that an individual can consume without reducing its availability to others and of which no one is deprived. These properties are defining features of pure public goods, examples of which include law enforcement, national defense, sewer systems, public parks, and the air we breathe. In reality, almost no good can satisfy the precise definitions of both non-rivalry and non-excludability~\cite{PublicGoods}, as these goods often suffer from over consumption \cite{ProblemPubGoods}, which leads to a decreased utility for consumers. This phenomena becomes more so during the Covid-19 pandemic, where social distance constraints are enforced so that only a limited number of people can share public spaces \cite{NYPubGoods}. A consequence of such constraints is that it results in completely closing parks or beaches \cite{WHOPub}, which leads to goods becoming non-public. Thus, on one hand public goods may no longer satisfy the non-rivalry and non-excludability properties associated with these resources while on the other hand the lack of regulation of public goods results in the consumers behaving in their own selfish interests when using these goods, which leads to the overuse of such shared resources, i.e., tragedy of the commons. In particular, Figure~\ref{CurrentScenario} illustrates these two contrasting outcomes and highlights the need for regulatory mechanisms that impose restrictions on usage of shared resources.

\begin{figure}[!h]
      \centering
      \includegraphics[width=0.85\linewidth]{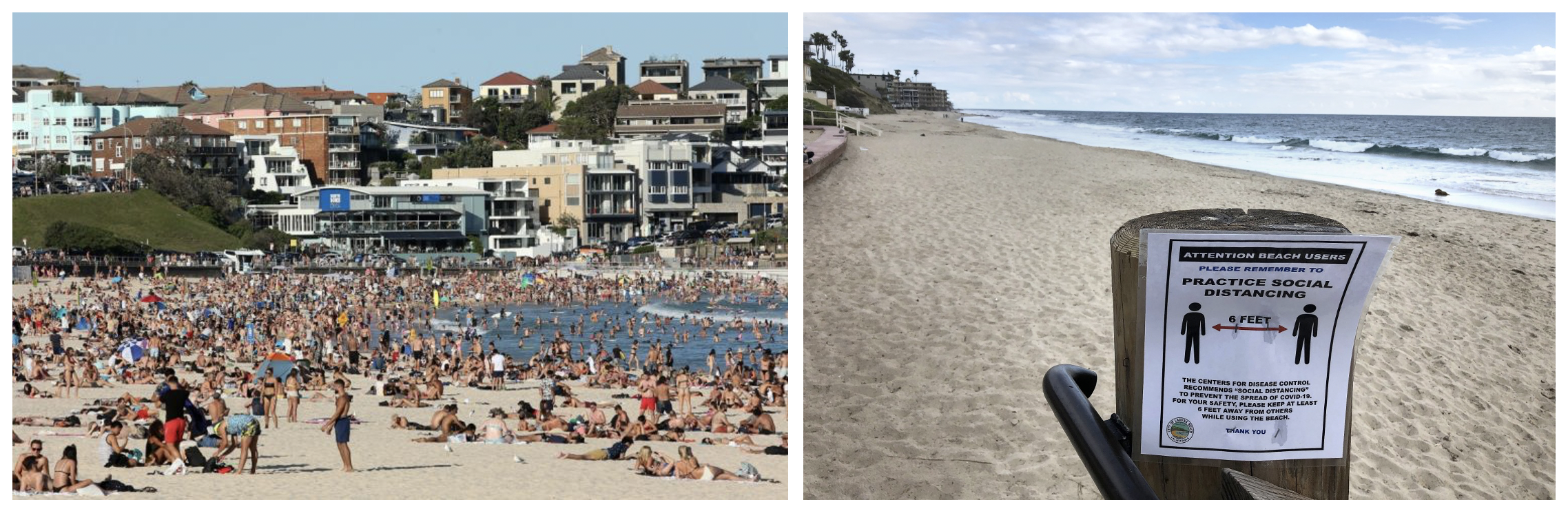}
      \caption{The current scenario involving either an overcrowded beach (left) or a completely unused beach (right) generating no value to society.}
      \label{CurrentScenario}
   \end{figure}
\vspace{-15pt}
In this paper, we attempt to achieve an intermediate between these opposing and undesirable outcomes through market mechanisms to efficiently allocate shared resources. To achieve such a balance, we study capacity constrained public resources, which helps prevent the overcrowding concerns associated with public goods and is in line with situations when strict capacity constraints in the use of public spaces are necessitated. In our work, we leverage the presence of alternatives in the case of many public goods, which enables us to distribute the consumer load over these resources. Examples of such alternatives include myriad local grocery stores and restaurants as well as multiple options for recreational use such as parks and beaches. As demand often outweighs supply of public goods, we need to make decisions on who gets preference to use capacity constrained public spaces. These allocation decisions are facilitated through a pricing mechanism that ensures the formation of a \textit{market equilibrium}, i.e., each agent purchases their most preferred bundle of goods affordable under the set prices. The pricing decisions must be made with fairness considerations, as public goods are by design available to all individuals and of which no person is deprived. We ensure fairness of our mechanism through two methods. First, we endow agents with artificial currency budgets and charge customers this currency for the use of public spaces. The non-monetary transfers ensure that allocation decisions are not biased towards those with higher incomes. Second, when making pricing decisions we simultaneously take into account individual consumer preferences, i.e., each agent's utilities for using public goods, while ensuring that resulting allocations are beneficial for society. 

To study our resource allocation problem under capacity constraints while considering both individual agent preferences and societal benefit, we resort to the canonical model of Fisher markets. In a Fisher market, consumers spend their budget of money (or artificial currency) to buy goods that maximize their individual utilities, while producers sell capacity constrained goods in exchange for currency. A key property of interest with Fisher markets is the formation of an equilibrium when the \textit{market clears}, i.e. all budgets are spent and all goods are sold. At this market equilibrium, buyers get their most preferred bundle of goods, while a social objective is maximized.

We first describe each agent's individual optimization problem in Fisher markets. In this framework, the decision variable for agent $i$ is the quantity of each good $j$ they wish to purchase and is represented by $x_{ij}$. We denote the allocation vector for agent $i$ as $\mathbf{x}_i \in \mathbb{R}^m$, when there are $m$ goods in the market. A key assumption of Fisher markets is that goods are divisible and so fractional allocations are possible, and so we interpret $x_{ij}$ as the probability that agent $i$ is allocated to good $j$. Finally, denoting $w_i$ as the budget of agent $i$, $u_i(\mathbf{x}_i)$ as the utility of agent $i$ as a concave function of their allocation and $\mathbf{p} \in \mathbb{R}_{\geq 0}^m$ as the vector of prices for the goods, individual decision making in Fisher markets can be modelled as the following optimization problem:
\begin{maxi!}|s|[2]                   
    {\mathbf{x}_i \in \mathbb{R}^m}                               
    {u_i(\mathbf{x}_i)  \label{eq:Fisher1}}   
    {\label{eq:Eg001}}             
    {}                                
    \addConstraint{\mathbf{p}^T \mathbf{x}_i}{\leq w_i \label{eq:Fishercon1}}    
    \addConstraint{\mathbf{x}_i}{\geq \mathbf{0}  \label{eq:Fishercon3}}  
\end{maxi!}
Here agent $i$ has a budget constraint~\eqref{eq:Fishercon1} and a non-negativity constraint~\eqref{eq:Fishercon3} on the allocations.

The vector of prices $\mathbf{p} \in \mathbb{R}_{\geq 0}^m$ that each agent observes are computed through the solution of a social optimization problem that aggregates the utilities of all agents. The social objective function used in the Fisher market literature is one wherein the optimal allocation maximizes the budget weighted geometric mean of the buyer's utilities. The choice of the social objective is such that under certain conditions on the utility function, there exists an equilibrium price vector defined as:
\begin{definition}
A vector $\mathbf{p} \in \mathbb{R}_{\geq 0}^{m}$ is an equilibrium price vector if $\sum_{i = 1}^{n} x_{ij}^*(p_j) = \Bar{s}_j$, $\forall j$, i.e., all the goods are sold, where each resource $j$ has a price $p_j$ and has a strict capacity constraint of $\Bar{s}_j \geq 0$, and $\sum_{j = 1}^{m} p_j x_{ij}^*(p_j) = w_i$, $\forall i$, i.e., budgets of all agents are completely used. Furthermore, $x^*(\mathbf{p})_i \in \mathbb{R}^{m}$ is an optimal solution of the individual optimization problem~\eqref{eq:Fisher1}-\eqref{eq:Fishercon3} for all agents $i$.
\end{definition}
This price vector is computed as the dual variables of the capacity constraint~\eqref{eq:FisherSocOpt1} in the following social optimization problem:
\begin{maxi!}|s|[2]                   
    {\mathbf{x}_i \in \mathbb{R}^m, \forall i \in [n]}                               
    {u(\mathbf{x}_1, ..., \mathbf{x}_n) = \sum_{i = 1}^{n} w_i \log( u_i(\mathbf{x}_i))  \label{eq:FisherSocOpt}}   
    {\label{eq:FisherExample1}}             
    {}                                
    \addConstraint{\sum_{i = 1}^{n} x_{ij}}{ = \Bar{s}_j, \forall j \in [m] \label{eq:FisherSocOpt1}}    
    \addConstraint{x_{ij}}{\geq 0, \forall i, j \label{eq:FisherSocOpt3}}  
\end{maxi!}
where there are $n$ agents in the economy and $m$ shared resources. We denote $[a]$ as the set $\small\{1, 2, ..., a \small\}$.

Since both the individual and social problems are convex optimization problems, the equivalence of their first order KKT conditions is necessary and sufficient at the equilibrium price condition. 
This establishes that under the prices set through the solution of the social optimization problem each agent receives their most favourable bundle of goods \cite{DevnaurPrimalDual}. 



While these appealing properties of Fisher markets have been leveraged with great success in applications including online advertising \cite{VaziApps} as well as in revenue optimization \cite{FisherApps}, there is a fundamental limitation with the Fisher market framework that limits its practical use in public goods allocation problems. This limitation stems from the consideration of only two types of constraints - budgets of buyers and capacities of goods - in the Fisher market model. In many public goods allocation problems, the availability of substitutes imposes additional physical constraints that are necessary to consider for the resulting allocation to be meaningful. To model the availability of public good alternatives, we pool together public goods serving similar functionality into their own resource \textit{types}. For instance, if we had multiple grocery stores and parks in a neighborhood, we could pool them into two different resource \textit{types}. We provide further examples of resource \textit{types} and delve into further detail on the constraints associated with these resource \textit{types} in sections~\ref{experiments} and~\ref{Applications}, where we elucidate the real world applications and practical implementation of this work. 

The presence of these additional constraints raises the question of whether we can still find appropriate market clearing prices while retaining the desirable properties of the Fisher market framework. This question leads us to the main focus of this paper, which is to:

\begin{center}
\textit{Design a market based mechanism that achieves the same properties as Fisher markets while also supporting additional physical constraints.}    
\end{center}

\subsection{Our Contributions} \label{Contributions}
\vspace{-3pt}
In our pursuit of such a mechanism, we start by defining each agent's individual optimization problem (\textbf{IOP}) with the addition of physical constraints. The properties of this problem \textbf{IOP} turn out to be fundamentally different from the traditional Fisher markets with linear utilities, as we no longer have guarantees on the existence and uniqueness of a market equilibrium. However, we derive a technical condition to overcome the question of the existence of a market equilibrium and provide a thorough characterization of the optimal solution of \textbf{IOP}.

Having established the existence of a market equilibrium, we then turn to deriving market clearing prices in the addition of physical constraints. We first show that market clearing conditions fail to hold under a natural extension of Fisher markets, wherein we add physical constraints to the social optimization problem~\eqref{eq:FisherSocOpt}-\eqref{eq:FisherSocOpt3} and derive prices using this constraint augmented problem (denoted as \textbf{SOP1}). This negative result points towards a formulation of a new social optimization problem (denoted as \textbf{BP-SOP}) wherein we perturb the budgets of agents by constants that depend on the dual variables of the physical constraints. We then show the following properties of \textbf{BP-SOP}:
\begin{itemize}
    \item \textbf{Market Clearing}: The budget perturbation constants are chosen such that the \textit{market clears} under the prices set corresponding to the dual variables of the capacity constraint of \textbf{BP-SOP}.
    
    \item \textbf{Economic Relevance of Solution}: A consequence of the market clearing property is that the prices are set appropriately so that each agent maximizes their utility subject to the set prices.
    
    \item \textbf{Computational Feasibility}: We present a fixed point scheme to determine the perturbation constants and establish its convergence through numerical experiments.
\end{itemize}

We note that the physical constraints we consider extend beyond public goods allocation, as such constraints arise on the buyer's side in retail, e-commerce and the AdWords market, as buyers have restrictions on the amount of goods they can purchase and advertisers on the number of people in each demographic class that they can target. In addition, such constraints help in achieving fairness by restricting the purchase of certain goods by individual agents to enable wider access. 

\vspace{-4pt}
\subsection{Related Work} \label{Literature}
\vspace{-3pt}
Setting appropriate market clearing prices has been a prominent topic of research at the intersection of economic and optimization theory. While Walras \cite{walras1954elements} was the first to question whether goods could be priced in a $n$ buyer $m$ good market such that each person receives a bundle of goods to maximize their utilities, it was Arrow and Debreu who established the existence of such a market equilibrium under mild conditions on the utility function of buyers \cite{arrow-debreu}. However, it was not until Fisher that there was an algorithm to compute equilibrium prices and the distribution of the $m$ goods amongst the $n$ buyers in the market \cite{Fisher-seminal}. Later Eisenberg and Gale formulated Fisher's original problem with linear utilities as a convex optimization problem that could be solved in polynomial time \cite{EisGale,Gale}. 

While Fisher markets have since been studied extensively in the computer science and algorithmic game theory communities, there has been recent interest in considering settings when additional constraints are added to the traditional Fisher market framework. For instance, Bei et al. \cite{bei-fisher} impose limits on how much sellers can earn and question the assumption that utilities of buyers strictly increase in the amount of the good allocated. A different generalization of Fisher markets is considered by Vazirani \cite{VaziApps}, Devanur et al. \cite{devnaur-spending} and Birnbaum et al. \cite{birnbaum-spending}, wherein utilities of buyers depend on prices of goods through spending constraints. Yet another generalization has been considered by Devanur et al. \cite{Devnaur-generalizations} in which goods can be left unsold as sellers declare an upper bound on the money they wish to earn and budgets can be left unused as buyers declare an upper bound on the utility they wish to derive. Along similar lines Chen et al. \cite{FisherPricing} study equilibrium properties when agents keep unused budget for use in the future. These generalizations of Fisher markets are primarily associated with spending constraints of buyers and earning constraints of sellers; however, to the best of our knowledge there has been no generalization of Fisher markets to the case of additional physical constraints, which are prevalent in public goods allocation problems as well as other applications. 

While such physical constraints have not been studied in the Fisher market literature, there have been other market equilibrium characterizations that take into account such constraints. One notable such work is that on the Combinatorial Assignment problem by Budish \cite{Budish} wherein a market based mechanism is used to assign students to courses while respecting student's schedule constraints. In Budish's framework, courses have strict capacity constraints and students are endowed with budgets and must submit their preferences to a centralized mechanism that provides approximately efficient allocations provided that budgets of students are not all exactly equal. We study the problem of allocating public goods under physical constraints from a different perspective than that considered by Budish. In particular, we approach the public goods allocation problem from the standpoint of setting equilibrium prices that are obtained through the maximization of a societal objective whilst also maximizing individual utilities subject to those prices.

Finally, since we are allocating public goods that are designed to be available to all individuals such that no one can be deprived of them, we must take the fairness of the resulting allocation into account. A popular method to achieve an equal playing ground for all agents has been the use of artificial currencies. For instance, Gorokh et al. established how artificial currencies can be equally distributed to agents to achieve fairness \cite{Banarjee-ArtCur}. We follow a similar idea in our work by endowing agents with artificial currency budgets that they can spend, which helps overcome concerns of priced mechanisms, e.g., congestion pricing, in regulating the use of public resources. 

The rest of this paper is organized as follows. We first present the individual optimization problem \textbf{IOP} and study properties of the corresponding market equilibrium in section~\ref{IOP-properties}. Then, in section~\ref{Impossibility}, we provide a motivation for why modifying Fisher markets is necessary to guarantee market clearing properties with the addition of physical constraints and propose a new budget perturbed social optimization problem that guarantees a market equilibrium. As the budget perturbed problem involves setting the perturbation constants as the dual variables of the added constraints, we present a fixed point iterative procedure to compute these constants in section~\ref{FixedPoint}. Finally, in section~\ref{Applications} we lay out a real world implementation plan of our mechanism and conclude the paper in section~\ref{Conclusion}. 

\section{Properties of the Individual Optimization Problem} \label{IOP-properties}

In this section, we study the individual optimization problem of agents in the addition of physical constraints that are not considered in the Fisher market framework. We start by defining a new individual optimization problem \textbf{IOP} in section~\ref{Model} and study properties regarding the existence and non-uniqueness of a market equilibrium in sections~\ref{counterexample-physical}-\ref{counterexample-uniqueness}. Finally, we close this section through a characterization of the optimal solution of \textbf{IOP} in section~\ref{IOP-Optimal}.

\subsection{Modelling Framework for Individual Optimization Problem} \label{Model}

As in Fisher markets, we model agents as utility maximizers and in this work, each agent's utility function is assumed to be linear in the allocations, which is a common utility function used in the Fisher market literature \cite{EisGale,LinUt}. We model the preference of an agent $i$ for one unit of good $j$ through the utility $u_{ij}$. Furthermore, we extend the Fisher market framework through the consideration of each agent's physical constraints. To model this physical constraint we consider each public good $j$ as belonging to exactly one resource \textit{type}, with the the set of all resource \textit{types} denoted as $T$. We further let $T_i \subseteq T$ denote the resource \textit{types} that each agent $i$ is interested in consuming and for the ease of exposition we normalize each agent's constraints such that each agent would like to obtain at most one unit of goods in each resource \textit{type}. These physical constraints are specified by a matrix $A^{(i)} \in \mathbb{R}^{m \times l_i}$, where $l_i = |T_i|$. Furthermore, the row corresponding to resource \textit{type} $t \in T_i$ is represented as a row vector $A_t^{(i)}$. Using our earlier defined notation to specify budgets and prices, we have the following individual optimization problem (\textbf{IOP}):
\begin{maxi!}|s|[2]                   
    {\mathbf{x}_i \in \mathbb{R}^m}                               
    {u_i(\mathbf{x}_i) = \sum_{j=1}^{m} u_{ij} x_{ij}  \label{eq:eq1}}   
    {\label{eq:Example001}}             
    {}                                
    \addConstraint{\mathbf{p}^T \mathbf{x}_i}{\leq w_i \label{eq:con1}}    
    \addConstraint{A_t^{(i)} \mathbf{x}_i}{\leq 1, \forall t \in T_i \label{eq:con2}}
    \addConstraint{\mathbf{x}_i}{\geq \mathbf{0}  \label{eq:con3}}  
\end{maxi!}
where we have a budget constraint~\eqref{eq:con1}, physical constraints~\eqref{eq:con2} and non-negativity constraint~\eqref{eq:con3}. We note that  $A_t^{(i)} \mathbf{x}_i \leq 1$ is identical to $\sum_{j \in t} x_{ij} \leq 1$, as each agent consumes at most one unit of goods in each resource \textit{type}. We now turn to studying the market equilibrium properties of \textbf{IOP}.



\subsection{Market Equilibrium may not Exist} \label{counterexample-physical}

In the traditional Fisher market framework with linear utilities, there exists a unique market equilibrium under mild assumptions \cite{algo-game-theory}. However, in the presence of additional physical constraints, as in the individual optimization problem~\eqref{eq:eq1}-\eqref{eq:con3}, an equilibrium price is not guaranteed to exist. In particular, we elucidate the non-existence of a market equilibrium through the following proposition.

\begin{restatable}{proposition}{nonexistenceeq}
\label{prop:non-existence-eq}
There exists a market wherein each good $j \in [m]$ has a potential buyer $i \in [n]$, i.e., $u_{ij}>0$, but no equilibrium for \textbf{IOP} exists.
\end{restatable}
To establish Proposition~\ref{prop:non-existence-eq}, we provide a counterexample presented in Appendix section~\ref{non-existence-prop1}.

\subsection{Condition to Guarantee Existence of Market Equilibrium} \label{guarantee-existence}

While Proposition~\ref{prop:non-existence-eq} indicates that in general we cannot expect a market equilibrium to exist for \textbf{IOP}, we now show that under a mild condition the market equilibrium is in fact guaranteed to exist. This condition arises from the fact that there may be instances, as in the case of the counterexample to prove Proposition~\ref{prop:non-existence-eq}, when agents cannot spend all of their budget. To guarantee that agents completely spend their budget, we must ensure that there is a good that is not restrained by physical constraints in the market so that agents can purchase more units of it to spend their entire budget. We formalize this notion by establishing the following theorem:

\begin{restatable}{theorem}{marketeq}
\label{thm:market-eq}
There exists a market equilibrium if for any agent $i$, there exists a good $j$, such that $j$ does not belong to any \textit{type}, i.e., it is not associated with any physical constraints, and $i$ has positive utility for all goods, i.e., $u_{ij}>0$, $\forall j$.
\end{restatable}

\begin{hproof}
We normalize the capacities of each good and the total budget of all agents to $1$, and consider an excess demand function $f_j(\mathbf{p}) = \sum_{i = 1}^{n} x_{ij}(p_j) - 1$ for $\mathbf{p} \in \Delta_m$, where $\Delta_m$ is a standard simplex. Next, we define a coloring function $c: \mathbf{p} \mapsto \small\{1, ..., m\small\}$, such that $c(\mathbf{p}) = j$ if $f_j(\mathbf{p}) \leq 0$ and $p_j \neq 0$. Such a coloring function on the standard simplex satisfies Sperner's lemma, which implies that we can find a $\mathbf{p}^*$, such that $f_j(\mathbf{p}^*) \leq 0$, $\forall j$, showing $\forall j$ that $\sum_{i = 1}^{n} x_{ij}(p_j^*) \leq 1$.

To prove that the above inequality is an equality, we suppose that $\exists j$, such that  $\sum_{i = 1}^{n} x_{ij}(p_j^*) < 1$. Then we consider two cases: i) $p_j >0$ and ii) $p_j = 0$. For both cases we find contradictions and prove the strict inequality is impossible under the condition that there exists a good $j$ without any physical constraints. This establishes our claim that $\mathbf{p}^*$ is the equilibrium price vector.
\end{hproof}

A detailed proof of this theorem is presented in Appendix section~\ref{existence-market-eq}. We also note the technical assumption is not very demanding. This is because we can achieve the requisite condition by introducing a good, such as "real" money that is available in sufficient quantity and for which each agent in the economy has a strictly positive value. Another method in which this can be achieved is by allowing agents to keep unused budget for use in the future, in which case we can treat budget as a good, which has been considered and analyzed in \cite{FisherPricing}.

\subsection{Market Equilibrium may not be Unique} \label{counterexample-uniqueness}

We now show that even if the market equilibrium exists, it may be that the equilibrium is not unique, which further establishes that the problem of determining a market equilibrium under the addition of physical constraints is fundamentally  different from the traditional Fisher market framework. In particular, we establish non-uniqueness of the market equilibrium through the following proposition.

\begin{restatable}{proposition}{nonuniquenesseq}
\label{prop:non-uniqueness-eq}
The market equilibrium for \textbf{IOP} may not be unique.
\end{restatable}

We use a counterexample presented in the Appendix section~\ref{non-uniqueness-prop2} to establish the non-uniqueness of the market equilibrium in the addition of physical constraints.

\subsection{Characterizing Optimal Solution of \textbf{IOP}} \label{IOP-Optimal}

Having established the conditions for the existence of a market equilibrium and that one cannot expect the resulting equilibrium to be unique, we now turn to characterizing the optimal solution of the \textbf{IOP}. In traditional Fisher markets, each agent purchases the goods $j^*$ corresponding to the highest \textit{bang-per-buck} ratio, i.e., agents purchase goods such that: $j^* = \argmax_j \left\{\frac{u_{ij}}{p_j} \right\}$. However, in the presence of physical constraints, when a buyer observes a price vector $\mathbf{p}$, which goods will they purchase in each resource \textit{type} and how many different goods will they purchase in each \textit{type}?

To answer these questions, we study the influence of the physical constraints on an agent's decision making problem through the consideration of a feasible solution set for buyer $i$ and resource \textit{type} $t \in T_i$ as follows:

\begin{definition}(Feasible Set). \label{def-sol-set}
Given a price vector $\mathbf{p} \in \mathbb{R}_{\geq 0}^{m}$, a feasible solution set for buyer $i$ and resource \textit{type} $t$ is given by:

$ S_{t} = \left\{(u_t, w_t) | \exists \left\{x_{ij} \right\}_{j \in t}, \sum_{j \in t} x_{ij} \leq 1, x_{ij} \geq 0, \forall j \in t, u_t = \sum_{j \in t} u_{ij} x_{ij}, w_t = \sum_{j \in t} x_{ij} p_j \right\} $
\end{definition}
The above definition specifies the utility and budget of agent $i$ when consuming \textit{type} $t$.

Next, we observe that the solution set $S_t$ can be viewed as lying in the convex hull of the points defined by $(u_{ij}, p_{ij})$, $j \in t$ and the origin in the price-utility plane, as shown by the enclosed region in Figure~\ref{virtual-products}. The lower frontier of this convex hull, as shown in bold, from the origin to $(u_{ij_{\max}}, p_{j_{\max}})$, where $j_{\max} = \argmax_{j \in t} \left\{u_{ij}\right\}$, is piece-wise linear and is characterised by slopes $\theta^t = (\theta_1^t, \theta_2^t, ..., \theta_{k_t}^t)$. As shown on the right in Figure~\ref{virtual-products}, given a fixed budget $w_t$ for \textit{type} $t$, the maximal utility that can be obtained from \textit{type} $t$ must be the intersection of the line $p = w_t$ and the lower frontier of the convex hull when $w_t \leq p_{j_{\max}}$. Otherwise the maximal utility obtained is $u_{ij_{\max}}$. Therefore, an optimal solution of \textbf{IOP} must lie on the lower frontier, with endpoints of the line segments corresponding to goods and line segments corresponding to \textit{virtual products}. Each \textit{virtual product} is characterized by its two endpoints, defined as:

\begin{definition} (Virtual Product).
A \textit{virtual product} is characterized by its two endpoints $A = (u_{ij_1}, p_{j_1})$ and $B = (u_{ij_2}, p_{j_2})$ with a slope $\theta_{j_1j_2} = \frac{p_{j_2} - p_{j_1}}{u_{ij_2} - u_{ij_1}}$. Then its \textit{bang-per-buck} $= \frac{1}{\theta_{j_1j_2}} = \frac{u_{ij_2} - u_{ij_1}}{p_{j_2} - p_{j_1}}$.
\end{definition}

\begin{figure}[!h]
      \centering
      \includegraphics[width=0.8\linewidth]{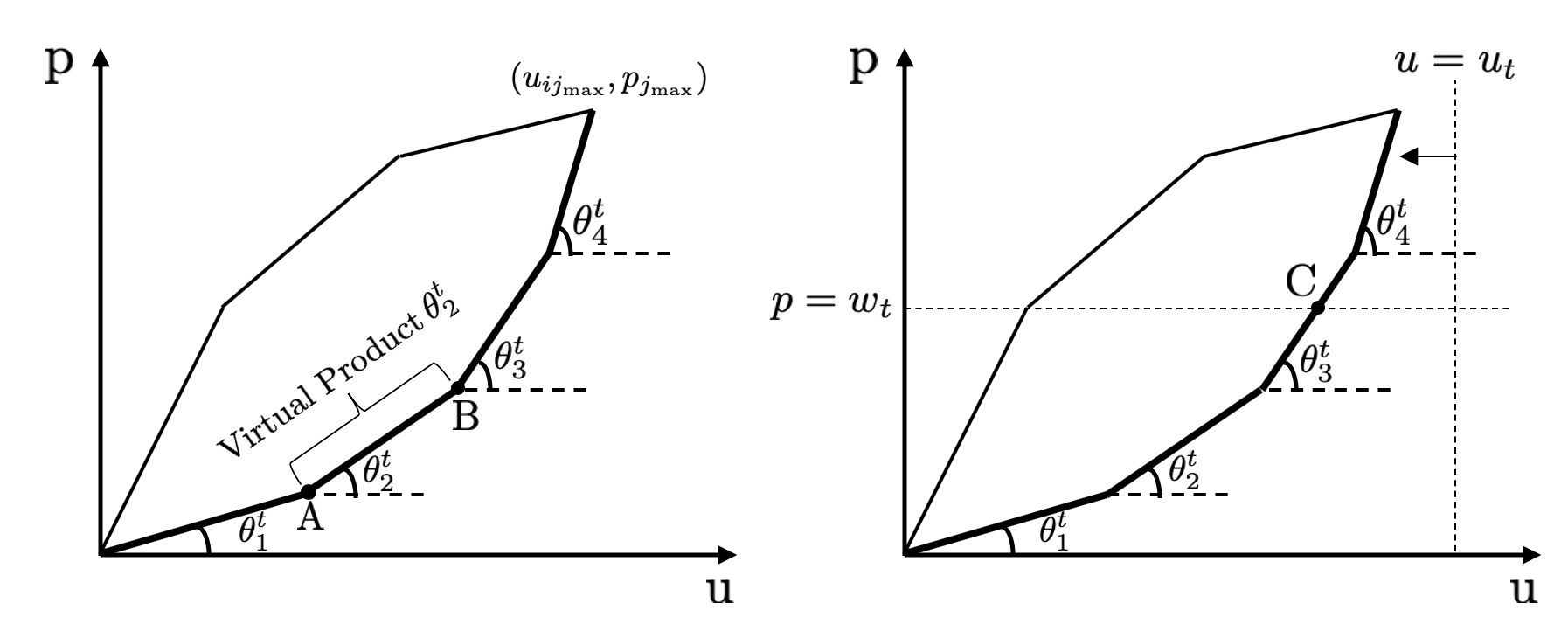}
      \caption{The enclosed region represents the convex hull corresponding to the solution set $S_t$. The vertices on the lower frontier (in bold) correspond to the goods and the segments correspond to \textit{virtual products}. The figure on the right shows that any optimal solution must lie on the lower frontier of the convex hull, as indicated by the point $C$.}
      \label{virtual-products}
   \end{figure}

While in Fisher markets, buyers purchase goods that maximize their \textit{bang-per-buck} ratios, in the presence of physical constraints, we show that agents purchase goods in the descending order of the \textit{virtual products}' \textit{bang-per-buck} ratios by establishing the following theorem:


\begin{restatable}{theorem}{iopopt}
\label{thm:iopoptimal}
Given a price vector $\mathbf{p} \in \mathbb{R}_{\geq 0}^{m}$, agent $i$ can obtain their optimal solution $\mathbf{x}_i^* \in \mathbb{R}^{m}$ of the \textbf{IOP} by mixing all \textit{virtual products} from different types together and spending their budget $w_i$ in the descending order of the \textit{virtual products}' bang-per-buck. Furthermore, at most one unit of each virtual product can be purchased by agent $i$.
\end{restatable}

Note that when two \textit{virtual products} have the same slope, irrespective of whether they are in the same \textit{type} or different \textit{types}, ties can be broken arbitrarily. Further, an immediate corollary follows, which answers the question of how many different goods an agent will purchase in each \textit{type}.

\begin{restatable}{corollary}{coroptsol}
\label{cor:coroptsol}
For any agent $i$, there exists an optimal solution $\mathbf{x}_i^* \in \mathbb{R}^{m}$, such that $i$ purchases two different goods in at most one resource type. For all other resource types, agent $i$ buys at most one good.
\end{restatable}

A detailed proof of Theorem~\ref{thm:iopoptimal} and Corollary~\ref{cor:coroptsol} are presented in Appendix sections~\ref{iop-optimal-pf} and~\ref{cor-optimal-pf} respectively. We note that the optimal solution $\mathbf{x}_i^*$ is the optimal solution for \textbf{IOP} given any price vector $\mathbf{p}$, which may not be the equilibrium price. The optimal solution may not be unique and Corollary~\ref{cor:coroptsol} states there exists such an optimal solution but this does not imply all solutions must satisfy these conditions. Furthermore, note that these results can easily be extended to the case when there exists some product that does not have any physical constraint. The detailed discussion and interpretation by examples of the optimal solution of the \textbf{IOP} are presented in Appendix sections~\ref{rmk-iopopt} and~\ref{examples-iopopt}. Having established properties of the \textbf{IOP}, we now turn to the problem of deriving equilibrium prices in the market in the addition of physical constraints.

\section{Generalizing the Fisher Social Optimization Problem to Accommodate Physical Constraints} \label{Impossibility}

A desirable property of a Fisher market is that the market equilibrium outcome maximizes a social objective while  individuals receive their most favoured bundle of goods given the set prices. As this property of Fisher markets holds when we consider budget and capacity constraints, a natural question to ask is whether we can still achieve this property under the addition of the physical constraints~\eqref{eq:con2}. We start by showing that the addition of these constraints and under no further modifications to Fisher markets, market clearing conditions fail to hold. To do this we define the social optimization problem (\textbf{SOP1}) with additional constraints in section~\ref{DefSocOpt} and then compare the KKT conditions of the two problems (\textbf{IOP} and \textbf{SOP1}) to establish that a market clearing outcome does not hold in section~\ref{KKTImposs}.

We then address this negative result by defining a perturbed social optimization problem (\textbf{BP-SOP}) in section~\ref{ReformSocOpt} in which we adjust the budgets of the agents. Then, in section~\ref{SocOptKKT}, we show how to choose these budget perturbations to guarantee the equivalence of its KKT conditions with that of the \textbf{IOP} when prices are set based on the dual variables of \textbf{BP-SOP}s capacity constraints. Finally, we provide an economic interpretation of our budget perturbed formulation in section~\ref{InterpSocOpt}.

\subsection{A Social Optimization Problem with Additional Constraints} \label{DefSocOpt}

We first define the natural extension of the Fisher market social optimization problem~\eqref{eq:FisherSocOpt}-\eqref{eq:FisherSocOpt3} with the addition of physical constraints~\eqref{eq:con2}. We note that barring the additional constraints, the following problem ~\eqref{eq:ImpossSocOpt}-\eqref{eq:ImpossSocOpt3} (\textbf{SOP1}) is identical to the Fisher market social optimization problem. 
\vspace{-15pt}
\begin{maxi!}|s|[2]                   
    {\mathbf{x}_i \in \mathbb{R}^m, \forall i \in [n]}                               
    {u(\mathbf{x}_1, ..., \mathbf{x}_n) = \sum_{i=1}^{n} w_i \log \left( \sum_{j=1}^{m} u_{ij}x_{ij} \right)  \label{eq:ImpossSocOpt}}   
    {\label{eq:ImpossExample1}}             
    {}                                
    \addConstraint{\sum_{i=1}^{n} x_{ij}}{ = \Bar{s}_j, \forall j \in [m] \label{eq:ImpossSocOpt1}}    
    \addConstraint{A_t^{(i)} \mathbf{x}_i}{\leq 1, \forall t \in T_i, \forall i \in [n] \label{eq:ImpossSocOpt2}}
    \addConstraint{x_{ij}}{\geq 0, \forall i, j \label{eq:ImpossSocOpt3}}  
\end{maxi!}

\subsection{A KKT Comparison of \textbf{IOP} and \textbf{SOP1}} \label{KKTImposs}

In the original Fisher Market formulation, the equilibrium price vector corresponds to the dual variables of the capacity constraints of the social optimization problem, and at this equilibrium price, the KKT conditions of the individual and social optimization problems are equivalent \cite{DevnaurPrimalDual}. We follow a similar approach when considering the individual and social optimization problems with additional constraints (\textbf{IOP} and \textbf{SOP1} respectively) and show that market clearing conditions fail to exist by establishing the following result.


\begin{restatable}{theorem}{imposs}
\label{thm:imposs}
The price vector $\mathbf{p} \in \mathbb{R}_{\geq 0}^{m}$ corresponding to the optimal dual variables of the capacity constraint~\eqref{eq:ImpossSocOpt1} of \textbf{SOP1} is not an equilibrium price, i.e., the market clearing KKT conditions of \textbf{IOP} and \textbf{SOP1} are not equivalent.
\end{restatable}

\begin{hproof}
We derive the first order necessary and sufficient KKT conditions of the social optimization problem \textbf{SOP1} and show that under the optimal price vector corresponding to the dual variables of the capacity constraint, the budgets of the agents will not be completely used up. As a result, a market clearing equilibrium cannot hold.
\end{hproof}
A detailed proof of this claim is presented in the Appendix section~\ref{Appendix1}. This result establishes that the prices a social planner would set through the solution of \textbf{SOP1} does not guarantee a market clearing outcome, as there would be agents with unused budgets. We now address this negative result through a modification to \textbf{SOP1} that will enable us to set market clearing prices under the addition of physical constraints.





\subsection{A Budget Perturbed Social Optimization Problem} \label{ReformSocOpt}

We present a reformulated social optimization problem in which we modify the budget of agents through a variable $\lambda_i$ for each agent $i$. This variable is introduced because of the additional physical constraints that are not present in the original Fisher market problem. The exact value that $\lambda_i$ should take so that the the market clears is derived in detail in the KKT analysis in section~\ref{SocOptKKT}. The Budget Perturbed Social Optimization Problem (\textbf{BP-SOP}) is represented as:
\vspace{-5pt}
\begin{maxi!}|s|[2]                   
    {\mathbf{x}_i \in \mathbb{R}^m, \forall i \in [n]}                               
    {u(\mathbf{x}_1, ..., \mathbf{x}_n) = \sum_{i=1}^{n} (w_i + \lambda_i) \log \left(\sum_{j = 1}^{m} u_{ij}x_{ij} \right)
    \label{eq:SocOpt}}   
    {\label{eq:Example01}}             
    {}                                
    \addConstraint{\sum_{i=1}^{n} x_{ij}}{ = \Bar{s}_j, \forall j \label{eq:SocOpt1}}    
    \addConstraint{A_t^{(i)} \mathbf{x}_i}{\leq 1, \forall t \in T_i, \forall i \in [n] \label{eq:SocOpt2}}
    \addConstraint{x_{ij}}{\geq 0, \forall i, j \label{eq:SocOpt3}}  
\end{maxi!}
with capacity constraints~\eqref{eq:SocOpt1}, physical constraints~\eqref{eq:SocOpt2} and non-negativity constraints~\eqref{eq:SocOpt3}. 



\subsection{Deriving Perturbation Constants Using KKT Conditions} \label{SocOptKKT}

We now show that under an appropriate choice of the $\lambda_i$ perturbations for all agents $i$, the KKT conditions of \textbf{BP-SOP} and \textbf{IOP} are equivalent when prices are set in the market through the dual variables of the capacity constraints~\eqref{eq:SocOpt1}. Observing that for any choice of $\lambda = (\lambda_1, ..., \lambda_n) $, \textbf{BP-SOP} remains a convex optimization problem, it is necessary and sufficient to verify the first order KKT conditions for \textbf{BP-SOP} and \textbf{IOP}. To establish the first order KKT equivalence between the two problems, we first define $r_{it}$ as the dual variable associated with the allocation constraint~\eqref{eq:SocOpt2} associated with agent $i$ and good \textit{type} $t$. Further, we define a \textit{fixed point} of the problem \textbf{BP-SOP} as one when $\lambda_i = \sum_{t = 1}^{l_i} r_{it}$, where $l_i = |T_i|$. The reasons for this choice of a \textit{fixed point} is its use in establishing the market equilibrium through the following theorem:

\begin{restatable}{theorem}{thmm}
\label{thm:thm2}
There is a one-to-one correspondence of the equilibrium price vector $\mathbf{p} \in \mathbb{R}_{\geq 0}^{m}$ and a fixed point solution of \textbf{BP-SOP}, i.e., $\lambda_i = \sum_{t = 1}^{l_i} r_{it}$, $\forall i$, where $r_{it}$ is the optimal dual multiplier of the constraint $A_t^{(i)} \mathbf{x}_i \leq 1$ in \textbf{BP-SOP}.
\end{restatable}

The above theorem states in one direction that the market clearing KKT conditions of \textbf{BP-SOP} are equivalent to that of the \textbf{IOP} if for each agent $i$, the budget perturbation constant $\lambda_i = \sum_{t = 1}^{l_i} r_{it}$. Furthermore, the above theorem also states the converse, i.e., any equilibrium price in the market for the social optimization problem \textbf{BP-SOP} must correspond to a \textit{fixed point}, i.e., $\lambda_i = \sum_{t = 1}^{l_i} r_{it}$. We now present a proof sketch of Theorem~\ref{thm:thm2} and show the complete derivation in Appendix section~\ref{proof-thm2}.

\begin{hproof}
We first derive the necessary and sufficient first order KKT conditions for the \textbf{BP-SOP} and \textbf{IOP}.
The forward direction of our claim follows from considering a market equilibrium of the \textbf{IOP} and using this to show that $\lambda_i = \sum_{t = 1}^{l_i} r_{it}$, $\forall i$ is the \textit{fixed point} of \textbf{BP-SOP}. For the converse, we can show that if we set $\lambda_i = \sum_{t = 1}^{l_i} r_{it}$, $\forall i$, then each agent completely uses up their budget, while all the goods are sold to capacity.
\end{hproof}


The above result implies that the prices a social planner would set through the solution of \textbf{BP-SOP} would guarantee a market clearing outcome. However, we note that for this mechanism to be implementable we need to determine the exact budget perturbation parameters which depend on the dual variables of the capacity constraint of \textbf{BP-SOP}, which we do not have knowledge of. This issue is addressed through a fixed point iterative procedure in section~\ref{FixedPoint}.

\subsection{Economic Relevance of Solution of \textbf{BP-SOP}} \label{InterpSocOpt}

Having established a one-to-one correspondence between the equilibrium price vector and the \textit{fixed point} solution of \textbf{BP-SOP}, we now show the economic relevance of the resulting allocations under the appropriately chosen budget perturbation constants. We first observe that due to the KKT equivalence of \textbf{BP-SOP} and \textbf{IOP} at the equilibrium price and the corresponding \textit{fixed point}, we have that each agent obtains their most preferred bundle of goods under the set prices. 

We now interpret the dual variable $r_{it}$ of the physical constraint as the price that agent $i$ must pay to purchase good type $t$. Hence, the total price that a buyer must pay to purchase goods $j$ belonging to type $t$ is $\sum_{j \in t} p_j x_{ij} + r_{it}$; however, the buyer only observes the price $p_j$ for good $j$ in the $\textbf{IOP}$. Thus, to reconcile the price difference the buyer observes and that in $\textbf{BP-SOP}$, we need to pay the additional price $\sum_{t} r_{it}$ for buying goods in the different \textit{types} by augmenting agents' budgets. Further, buyers are no longer purchasing goods with the highest \textit{bang-per-buck}, and under the adjusted price set $p_j' = p_j + r_{it}$, where one unit of good $j$ is purchased and $j \in t$, agents are purchasing goods with the highest "adjusted" \textit{bang-per-buck}. Finally, we observe that more constrained agents have larger weights $\lambda_i$ than less constrained agents, ensuring more constrained agents have "higher priorities" and thus an allocation that lies within their feasible constraint set. 



\section{Fixed Point Scheme to determine Perturbation Constants} \label{FixedPoint}

In the above budget perturbed social optimization formulation, we required that $\lambda_i = \sum_{t = 1}^{l_i} r_{it}$, i.e., $\lambda_i$ depends on the dual variables of the problem, which we have no knowledge of apriori. 
In this section, we show how to compute the appropriate value of $\lambda_i$ through the means of a fixed point iteration in section~\ref{algorithm-fixed-pt} and numerically establish its convergence through experiments in section~\ref{experiments}. 

\subsection{Fixed Point Iteration Algorithm} \label{algorithm-fixed-pt}

To determine the true value of the perturbation parameters specified by the vector $\lambda \in \mathbb{R}_{\geq 0}^{n}$, we consider an iterative scheme of the form $G\left(\lambda_1^{(k)}, ..., \lambda_n^{(k)}\right) = \left(\mathbf{r}_{1}^{(k)}, ..., \mathbf{r}_n^{(k)}\right)$, where we update our perturbation parameters as: $\left(\lambda_1^{(k+1)}, ..., \lambda_n^{(k+1)}\right) = \left(\sum_{t = 1}^{l_i} r_{1t}^{(k)}, ..., \sum_{t = 1}^{l_i} r_{nt}^{(k)} \right)$. Here $G$ is a function that takes in the $k^{th}$ iterate $\lambda_i^{(k)}$ for all agents $i$, solves the corresponding social optimization problem \textbf{BP-SOP} and returns the dual variables, $\mathbf{r}_i^{(k)} \in \mathbb{R}_{\geq 0}^{l_i}$, corresponding to the physical constraints.

The following algorithm depicts the fixed point iterative scheme, where, $\lambda = (\lambda_1, ..., \lambda_{n})$, and $\mathbf{R} = (\mathbf{r}_1, ..., \mathbf{r}_n)$, the dual variables corresponding to the physical constraints in \textbf{BP-SOP}.
\begin{algorithm} 
\label{alg:Algo1}
\SetAlgoLined
\SetKwInOut{Input}{Input}\SetKwInOut{Output}{Output}
\Input{Tolerance $\epsilon$, Function $G(\cdot)$ to calculate dual variables of physical constraints of \textbf{BP-SOP}}
\Output{Budget Perturbation Parameters $\lambda$}
$\lambda \leftarrow \mathbf{0} $ \;
$\mathbf{R} \leftarrow G(\lambda)$ \;
$q_i \leftarrow \sum_{t = 1}^{l_i} r_{it}$, $\forall i$ \;
 \While{$\norm{\lambda - \mathbf{q}}_2 > \epsilon$}{
  $\lambda_i \leftarrow \sum_{t = 1}^{l_i} r_{it}$, $\forall i$ \;
  $\mathbf{R} \leftarrow G(\lambda)$ \;
  $q_i \leftarrow \sum_{t = 1}^{l_i} r_{it}$, $\forall i$ \;}
\caption{Fixed Point Scheme}
\end{algorithm}
\subsection{Numerical Experiments with Iterative Scheme} \label{experiments}

We now numerically evaluate the convergence of the iterative scheme for \textbf{BP-SOP} through the implementation of Algorithm 1 to the allocate agents to public spaces.

We consider a neighborhood with a population of $n=200$ people and $m = 6$ public spaces, with three resource \textit{types} including two grocery stores, two parks, and two beaches. The capacities of the public spaces are $\Bar{s}_j = 100$, $\forall j \in [m]$ and the physical constraint used in this numerical experiment is that each individual would not want to go to more than one of the public spaces within the same resource \textit{type} over the course of a day, defining three identical physical constraints for each person. 

Further, we consider an economy where each person $i$ is endowed with a budget $w_i$ through electronic coupons that they can spend over a time horizon of a day to use these public spaces. In this experiment, we endow agents with random budgets, although these could be assigned differently including giving each person equal budgets. Finally, the people have different values for the public spaces, and their preferences are captured through their utilities for availing each of these public spaces. This can depend on a range of factors including their proximity to the public space and quality of service of a given public space within a certain \textit{type} relative to other public spaces within the same \textit{type}. For the purposes of this analysis, the utilities $u_{ij}$, $\forall i, j$, are randomly generated.

On the above defined problem instance, we run Algorithm 1, wherein we terminate when $\norm{\lambda^{(k)} - \sum_{t = 1}^{3} \mathbf{r}_t^{(k)}} \leq \epsilon$, where $\lambda^{(k)} = \left(\lambda_1^{(k)}, ..., \lambda_{n}^{(k)}\right)$ and $\mathbf{r}_t^{(k)} = \left(r_{1t}^{(k)}, ..., r_{nt}^{(k)}\right)$, where $r_{it}^{(k)}$ is the dual variable of the optimization problem at iteration $k$, and $n$ is the total number of people, which is 200 for our problem instance. The experiment confirmed that for the above public goods allocation problem (as well as for other tested problem instances) that the iterative scheme converges in fewer than 40 iterations to a fixed point as can be observed in Figure~\ref{BudgetUtilityPerturbed}, highlighting the computational feasibility of our mechanism for large problem instances. Furthermore, we did not require any structure on the budgets and utilities, as these were randomly chosen, to observe convergence.

\begin{figure}[!h]
      \centering
      \includegraphics[width=0.7\linewidth]{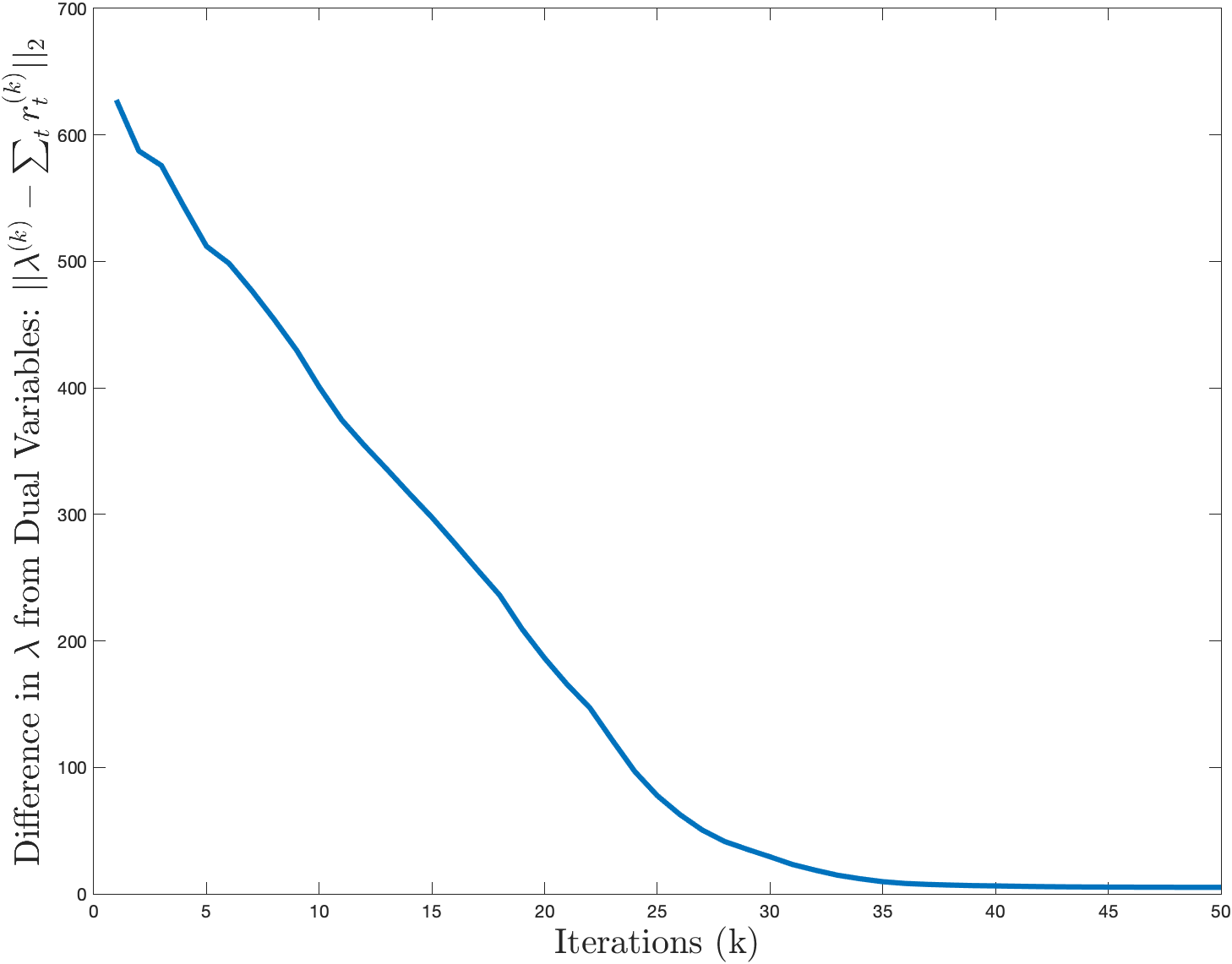}
      \caption{Numerical Convergence of fixed point iterative scheme with 200 people, 6 goods and 3 \textit{types} of goods. The budgets and utilities were assigned randomly to agents, and the convergence is observed after about 40 iterations.}
      \label{BudgetUtilityPerturbed}
   \end{figure}
We note that the numerical experiments confirmed the feasibility of the allocations with each good being used to capacity. In particular, for each agent the sum of their allocations within each resource \textit{type}, i.e., grocery stores, parks and beaches, is exactly one, indicating feasibility with respect to the added physical constraints. We note that in the absence of the physical constraint, if we used the pure Fisher market framework then an individual could have received a highly undesirable allocation. For instance, the pure Fisher market framework without these physical constraints may allocate a certain agent to only grocery stores, if the agent had the highest utility to price ratio for grocery stores as compared to other public resources. Our mechanism overcomes this problem by accounting for these physical constraints whilst guaranteeing a socially beneficial allocation.


\section{Implementation of Mechanism} \label{Applications}

In this section, we demonstrate how our proposed budget adjustment mechanism could be implemented in the real world. In particular, we elucidate the details of the implementation of our mechanism to efficiently allocate "time of use" permits for a public space, such as a beach. Our public goods are the different blocks of time that people can use the beach with the physical constraint that people can use the beach at most once per day, i.e., the resource \textit{types} are the days an agent is interested in using the beach. We first highlight the broad implementation plan of our mechanism in setting prices to allocate "time of use" permits to people using a beach:
\begin{enumerate}
    \item As a first step, we will need to create a schedule for the usage of the beach, i.e., we need to divide the day into different blocks of time, each of which corresponds to a public good.
    \item We will set non-monetary prices for the different times of use through our market mechanism.
    \item Finally, the budgets of individuals can be allocated through electronic coupons, which can then be used to purchase "time of use" permits priced in non-monetary units.
\end{enumerate}

However, in real world implementation, there is an additional hurdle in setting appropriate prices, as the market does not have information on the utilities of buyers, which serve as inputs to \textbf{BP-SOP}. Thus, we propose the following methodology to set market clearing prices using our mechanism while learning information on the utilities of buyers over time:

\begin{enumerate}
    \item We classify people into different groups based on parameters such as their income level and demographic information, so that we can assign all members within a group the same utility coefficients. This helps in reducing the computational complexity of the problem.
    \item We initialize a price for the different "times of use" by running our market mechanism with some prior on the utilities of the agents, which can be obtained through data on prior customer demands or just be assigned uniformly amongst all agent groups for all goods.
    \item Based on the set prices, people will purchase "time of use" permits. We will observe the "buying-behavior" of these different groups of individuals through a comparison between the demand of people at the set prices and the supply of these goods, i.e., the good capacities. 
    \item This observed discrepancy between the capacities of the beach and the corresponding demand can be used to re-calibrate the utility coefficients (e.g., through non-parametric statistical estimation of observed demands) and adjust prices in the market.
    \item We can continue this utility and price adjustment process iteratively until we have that the capacities of the beach at each "time of use" closely matches the demand of the good.
\end{enumerate}

Through this adaptive learning process, we will be able to learn each agent's utility coefficients and set appropriate prices (evaluated through the dual variables of our proposed Fisher market mechanism) to achieve the desired market clearing outcome. In particular, employing "time of use" permits (at three hour intervals) that are priced according to our mechanism will help achieve an intermediate outcome between the two extremes of an overused and an underused beach presented in Figure~\ref{CurrentScenario}. We depict such an intermediate market outcome in Figure~\ref{TimeOfUse}.
\vspace{-15pt}
\begin{figure}[!h]
      \centering
      \includegraphics[width=0.9\linewidth]{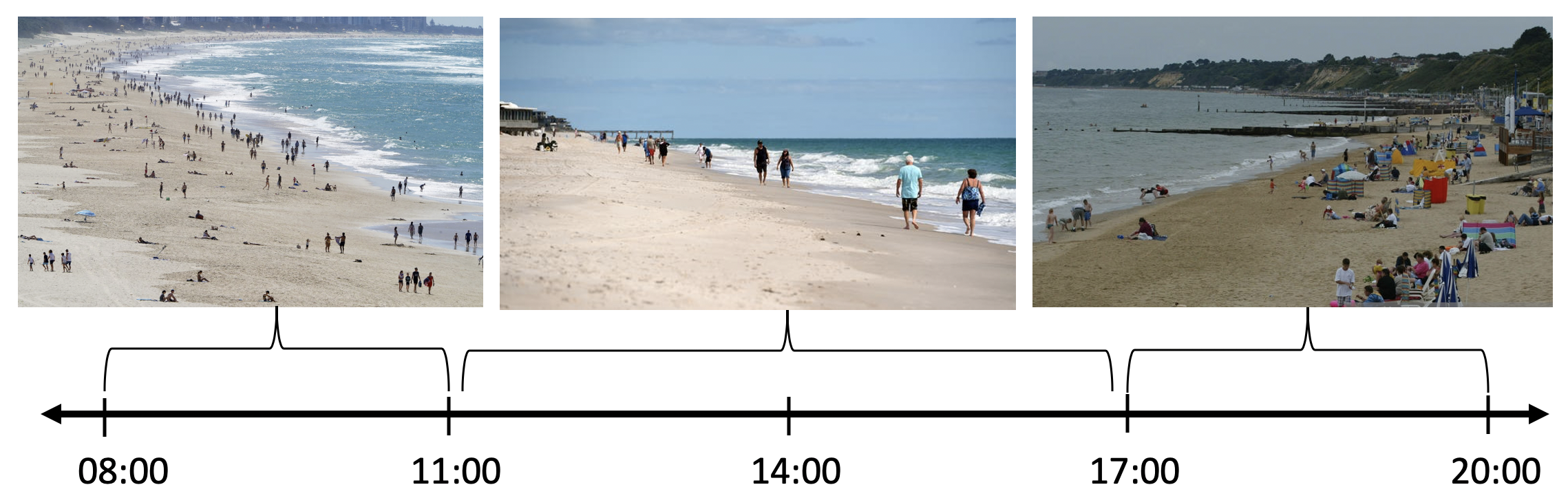}
      \caption{Our market outcome that will achieve an intermediate between an overcrowded and an underused beach. Through our pricing mechanism, customers will purchase time of use permits (in this case at 3 hour intervals) in a “controlled” manner resulting in neither overcrowded or underused public resources.}
      \label{TimeOfUse}
   \end{figure}
\vspace{-20pt}
\section{Conclusions and Future Work} \label{Conclusion}

In this work, we have developed market based mechanisms to more efficiently allocate capacity constrained public goods that are priced in non-monetary units. We defined a new individual optimization problem \textbf{IOP} in the presence of physical constraints and established market equilibrium properties of this problem, including existence, non-uniqueness and thoroughly characterized its optimal solution. Even though the properties of \textbf{IOP} are fundamentally different from that of the individual optimization problem in Fisher markets, we proposed a mechanism to derive a market equilibrium in the presence of physical constraints, thereby generalizing the Fisher market framework. In particular, we reformulated the Fisher market setup to account for additional physical constraints by perturbing the budgets of agents and defining a new social optimization problem \textbf{BP-SOP}. We then showed a one-to-one correspondence between the equilibrium price vector and a \textit{fixed point} solution of \textbf{BP-SOP} through the verification of KKT conditions. Next, we established the significance of the budget perturbation constants and that under the appropriate choice of these constants, a market equilibrium is attained such that under the set prices each agent's individual utilities are maximized. To obtain the right budget perturbation parameters, we used a fixed point iteration scheme for the reformulated social optimization problem and numerically showed the convergence of this iterative procedure indicating the applicability of our mechanism for real world problem instances.

There are various interesting directions of future research that warrant more study. First, the allocations provided by our market mechanism are fractional, and as we would like to make discrete allocations, it would be interesting to investigate the loss in social efficiency under integral constraints. Furthermore, while we have numerically shown convergence guarantees of our iterative scheme, it would be beneficial to theoretically understand the convergence to the fixed point as well as the rate of convergence of our procedure. Next, we believe that a stronger characterization of the computational complexity of this problem would provide a more nuanced appreciation of whether computing the exact fixed point is feasible in polynomial time. Finally, an interesting area of research is generalizing this framework to an online setting in which customers arrive in the market platform sequentially and an irrevocable decision needs to be made about the prices that should be set in the market while still achieving a socially efficient allocation.

\section*{Acknowledgements}

This work is partly supported by the Research Grant Council of Hong Kong (GRF Project no. 16215717 and 16243516), National Science Foundation (NSF) under CAREER Award CMMI1454737 and Toyota Research Institute (TRI). This article solely reflects the opinions of its authors and not NSF, TRI, or any other entity.

\bibliographystyle{unsrt}
\bibliography{main}

\clearpage

\section{APPENDIX}

\subsection{Proof of Proposition 1} \label{non-existence-prop1}

We now restate and proceed to prove Proposition~\ref{prop:non-existence-eq}.

\nonexistenceeq*

\begin{proof}
Consider a market with two buyers and two goods, with capacities $\Bar{s}_1 = 1.5$ and $\Bar{s}_2 = 0.5$, and utilities and budgets as specified in Table 1. 

\begin{table}[tbh]
\caption{Utilities and Budgets of two Buyers in a two buyer, two good market}
\centering
\begin{tabular}{|c|c|c|c|}
\hline
                 & \textbf{Utility for Good 1} & \textbf{Utility for Good 2} & \textbf{Budget} \\ \hline
\textbf{Buyer 1} & 200                         & 0.1                         & 15              \\ \hline
\textbf{Buyer 2} & 100                         & 1.1                         & 5               \\ \hline
\end{tabular}
\end{table}
Denoting $x_{ij}$ as the amount of good $j$ that buyer $i$ purchases, we consider the following physical constraints: $x_{11} + x_{12} \leq 1$ and $x_{21} + x_{22} \leq 1$ for buyer's 1 and 2 respectively.

To show that no equilibrium exists in this market, we consider the following two cases:

\begin{enumerate}
    \item $p_1 > 10$: From the physical constraint, buyer 1 can buy at most one unit of good 1 and buyer 2 cannot afford 0.5 units of good 1 and so it follows that $\sum_{i = 1}^{2} x_{ij}(p_1) < \Bar{s}_1 = 1.5$, establishing that good 1 cannot be cleared. Thus, $p_1>10$ cannot be a market clearing price.
    \item $p_1 \leq 10$: Since the utilities of both buyers is higher for good 1 than for good 2, it must be that the price for good 1 must be higher than the price for good 2. However, this implies that $p_2 \leq p_1 \leq 10$, which implies that buyer 1 cannot use up their budget, since the buyer can purchase at most one unit of both goods combined. Thus, $p_1 \leq 10$ cannot be a market clearing price.
\end{enumerate}

The two cases above establish that for this example a market equilibrium fails to exist in the presence of additional physical constraints not considered in the Fisher market framework.
\end{proof}

\subsection{Proof of Theorem 1} \label{existence-market-eq}

We now restate and proceed to prove Theorem~\ref{thm:market-eq}.

\marketeq*

\begin{proof}
We start by normalizing the capacities of each good to $1$, i.e., $\Bar{s}_j = 1$ $\forall j$, and normalizing the total budget of all agents to 1, i.e., $\sum_{i = 1}^{n} w_i = 1$. 

Next, we define an excess demand function $f_j(\mathbf{p}^*) = \sum_{i = 1}^{n} x_{ij}(p_j^*) - 1$ and use Sperner's lemma to prove the existence of an equilibrium price vector $\mathbf{p}^* = (p_1^*, p_2^*, ..., p_m^*) \in \Delta_m$, where $\Delta_m$ is the standard simplex, such that for all goods $j$, the excess demand function $f_j(\mathbf{p}^*) = \sum_{i = 1}^{n} x_{ij}(p_j^*) - 1 = 0$.

Next, for any $\mathbf{p} \in \Delta_m$, we define a coloring function $c: \mathbf{p} \mapsto \small\{1, ..., m\small\}$, such that $c(\mathbf{p}) = j$ if $f_j(\mathbf{p}) \leq 0$ and $p_j \neq 0$. We claim that such a coloring satisfies Sperner's lemma.

To see this, we first note that all the $m$ corner points of the simplex must be colored with different colors, as at each of the corner points $k$ exactly one entry $p_k = 1 \neq 0$. In fact when $p_k = 1$, as the budgets are normalized to $1$, it must be that $f_k(\mathbf{p}) \leq 0$, as at most one unit of good $k$ can be purchased by all agents collectively. Thus, we have that all the corner points of the simplex must be colored with a different color. For any other bounding vertex there must exist $j$, such that $p_j \neq 0$ and $f_j(\mathbf{p}) \leq 0$. This is because if for all $j$, such that $p_j \neq 0$, we have that $f_j(\mathbf{p})>0$ then $\sum_{i = 1}^{n} \sum_{j = 1}^{m} p_j x_{ij} >1 = \sum_{i = 1}^{n} w_i$, a contradiction. Thus, we have a valid coloring, which satisfies the condition of Sperner's lemma.


Next, by Sperner's lemma we have that there must exist a base simplex such that each of its corner points have distinct colors. Taking finer and finer triangulations and using that the demand function is continuous, we can find a $\mathbf{p}^*$, such that $f_j(\mathbf{p}^*) \leq 0$, $\forall j$, as each subsequent triangulation still lies on the standard simplex. Thus, we have shown $\forall j$ that $\sum_{i = 1}^{n} x_{ij}(p_j^*) \leq 1$.

Now to prove that the above inequality is in fact an equality, we proceed by contradiction. In particular, we suppose that if $\exists j$ with $p_j>0$, such that $\sum_{i = 1}^{n} x_{ij}(p_j^*) < 1$, then clearly there must exist some buyer $i$ whose budget was not used up since $\sum_{i = 1}^{n} \sum_{j = 1}^{m} p_j x_{ij} < 1 = \sum_{i = 1}^{n} w_i$. However, by the condition that for all agents $i$ there exists a good $j$ without physical constraints, $i$ can buy more units of good $j$ (and gain a strictly positive utility), giving us our desired contradiction. Thus, $\mathbf{p}^*$ is an equilibrium price vector.

Next, suppose that $\exists j$ with $p_j = 0$, such that $\sum_{i = 1}^{n} x_{ij}(p_j^*) < 1$ and assume that $j \in t$, i.e., good $j$ belongs to good \textit{type} $t$. Then we must have that for all agents $i$, the physical constraint $\sum_{j \in t} x_{ij} = 1$, i.e., it is tight, as otherwise $i$ can buy more of good $j$ at $p_j = 0$ to increase his/her utility. We assume that \textit{type} $t$ has $k$ products and the physical constraints for type $t$ are such that $x_{ij} + \sum_{h \in t, h \neq j} x_{ih} \leq t_i$ for all agents $i$, where $t_i$ must be such that $\sum_{i = 1}^{n} t_i \geq k$ to to ensure a feasible solution for the problem. This is because otherwise there must exist some product that cannot be sold out. Now, since physical constraints for all agents for good \textit{type} $t$ are tight, we have that $\sum_{i = 1}^{n} \left( x_{ij} + \sum_{h \in t, h \neq j} x_{ih} \right) = \sum_{i = 1}^{n} t_i \geq k$. However, since \textit{type} $t$ has $k$ products it must be that $\sum_{i = 1}^{n} \left( x_{ij} + \sum_{h \in t, h \neq j} x_{ih} \right) = k$, which implies that the product is sold to capacity. Since we normalized the capacities of all products to one unit, it follows that $\sum_{i = 1}^{n} x_{ij}(p_j^*) = 1$, a contradiction. This again establishes that $\mathbf{p}^*$ is an equilibrium price vector.


Having considered both cases of $p_j = 0$ and $p_j>0$, we have established that for all $j$ that $\sum_{i = 1}^{n} x_{ij}(p_j^*) = 1$, with each agent's budget being completely used up, establishing that $\mathbf{p}^*$ is the equilibrium price vector, and so a market equilibrium exists.
\end{proof}

\subsection{Proof of Proposition~\ref{prop:non-uniqueness-eq}} \label{non-uniqueness-prop2}

We now restate and proceed to prove Proposition~\ref{prop:non-uniqueness-eq}.

\nonuniquenesseq*

\begin{proof}
To establish that the market equilibrium may not be unique, we consider the example of three buyers and three goods, with capacities $\Bar{s}_1 = 1, \Bar{s}_2 = 2$ and $\Bar{s}_3 = 1$ respectively and utilities and budgets as specified in Table 2. Next, denoting $x_{ij}$ as the amount of good $j$ that buyer $i$ purchases, we consider the following physical constraints: $x_{11} + x_{12} \leq 1$, $x_{21} + x_{22} \leq 1$ and $x_{31}+x_{32} \leq 1$ for buyer's 1, 2 and 3 respectively, with buyers 1, 2 and 3 being able to purchase any amount of good 3.

\begin{table}[tbh] \label{table-non-unique}
\caption{Utilities and Budgets of three Buyers in a three buyer, three good market}
\centering
\begin{tabular}{|c|c|c|c|c|}
\hline
                 & \textbf{Utility for Good 1} & \textbf{Utility for Good 2} & \textbf{Utility for Good 3} & \textbf{Budget} \\ \hline
\textbf{Buyer 1} & 100                         & 1  & 2                       & 20              \\ \hline
\textbf{Buyer 2} & 1                         & 100    & 1                     & 10               \\ \hline
\textbf{Buyer 3} & 1                         & 100     & 1                    & 10               \\ \hline
\end{tabular}
\end{table}

We observe that if we set prices $\mathbf{p} = (11, 10, 9)$ or $\mathbf{p} = (10, 10, 10)$ that the resulting allocation will be $x_{11} = 1$, $x_{13} = 1$, $x_{22} = 1$ and $x_{32} = 1$. We note that under these prices and allocations, each agent's budgets are completely used up and each of the goods are sold to capacity, establishing that both these price vectors are equilibrium price vectors. Thus, we have shown that the equilibrium market price may not necessarily be unique.
\end{proof}

\subsection{Proof of Theorem~\ref{thm:iopoptimal}} \label{iop-optimal-pf}

We now restate and proceed to prove Theorem~\ref{thm:iopoptimal}.

\iopopt*

\begin{proof}
From Definition~\ref{def-sol-set}, we know for all resource \textit{types} $t$ that $u_t = \sum_{j \in t} u_{ij} x_{ij}^*, w_t = \sum_{j \in t} x_{ij}^* p_j$.

Since $\mathbf{x}_i^*$ is an optimal solution, it must be that $(\mathbf{x}_i^*)_{j \in t}$, i.e., the restriction of $\mathbf{x}_i^*$ to the goods in resource \textit{type} $t$, is an optimal solution to the following problem:

\begin{maxi!}|s|[2]                   
    {\left\{x_{ij} \right\}_{j \in t}}                               
    {\sum_{j \in t} u_{ij} x_{ij}  \label{eq:neweq1}}   
    {\label{eq:newExample001}}             
    {}                                
    \addConstraint{\sum_{j \in t} x_{ij} p_j}{\leq w_t \label{eq:newcon1}}    
    \addConstraint{\sum_{j \in t} x_{ij}}{\leq 1, \label{eq:newcon2}}
    \addConstraint{x_{ij}}{\geq 0, \forall j \in t  \label{eq:newcon3}}  
\end{maxi!}

\end{proof}

The optimal solution $\left\{x_{ij}^* \right\}_{j \in t}$ corresponds to some point $C$ (as depicted in Figure \ref{virtual-products}) in the lower frontier of the convex hull $S_t$ from $(0, 0)$ to $(u_{ij_{\max}}, p_{j_{\max}})$, where $j_{\max} = \argmax_{j \in t} u_{ij}$. Thus, such a solution $\left\{x_{ij}^* \right\}_{j \in t}$ can be viewed as purchasing \textit{virtual products} given by the different slopes $\theta$ in the descending order of their \textit{bang-per-buck}, which is equivalent to the ascending order of the slopes $\theta$. In particular, one needs to purchase \textit{virtual products} in the descending order from the origin to the point $C$ (as depicted in Figure \ref{virtual-products}).

Having established that agent $i$ purchases goods in the descending order of their \textit{bang-per-buck} within each good \textit{type}, we now show that agent $i$ in fact purchases goods in the descending order of their \textit{bang-per-buck} amongst all good \textit{types}. Thus, consider two different good types, $t$, $t'$ and suppose for contradiction that agent $i$ does not purchase goods in the descending order of their \textit{bang-per-buck}. Thus, it must be that $\exists$ $\theta_a \in t$ and $\theta_b \in t'$, such that $\theta_a <\theta_b$, where some amount of $\theta_b$ is bought but $\theta_a$ is not completely purchased, i.e., strictly less than one unit of it is bought.

Without loss of generality, assume that $\theta_a$ is defined by two goods $j_1, j_2 \in t$ (with $u_{ij_1} < u_{ij_2}$) and $\theta_b$ is defined by two goods $j_3, j_4 \in t'$ (with $u_{ij_3} < u_{ij_4}$), as depicted in Figure~\ref{thm2-proof-opt-fig}. Next, suppose that we have an optimal solution for these four goods given by $(x_{ij_1}^*, x_{ij_2}^*, x_{ij_3}^*, x_{ij_4}^*)$.

\begin{figure}[!h]
      \centering
      \includegraphics[width=0.9\linewidth]{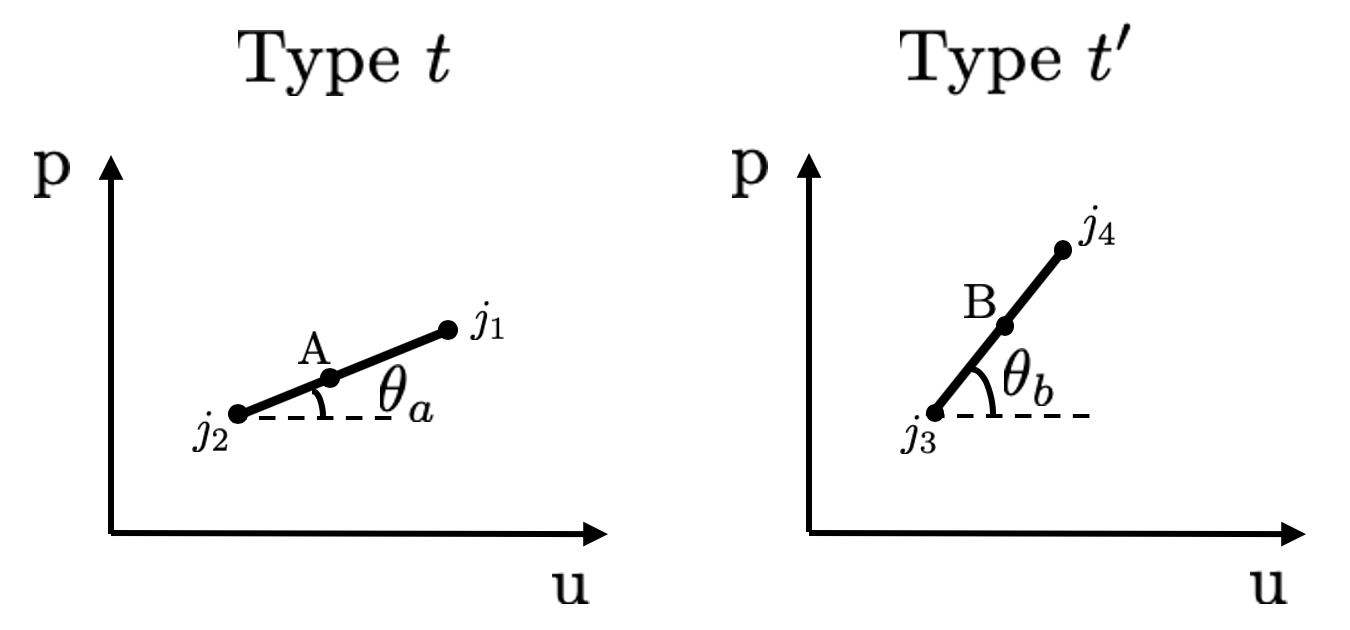}
      \caption{Goods $j_1$ and $j_2$ belong to \textit{type} $t$ and goods $j_3$ and $j_4$ belong to \textit{type} $t'$. The amount of \textit{virtual product} consumed lies somewhere along the line segment denoted by $A$ and $B$ for resource \textit{types} $t$ and $t'$ respectively. We further have that $\theta_a<\theta_b$.}
      \label{thm2-proof-opt-fig}
   \end{figure}

Next, let $x_{ij_4}^{'} = x_{ij_4}^* - \epsilon$ and $x_{ij_3}^{'} = x_{ij_3}^* + \epsilon$. We now observe that since $(x_{ij_3}^*, x_{ij_4}^*)$ is feasible with respect to the physical constraint, so is $(x_{ij_3}^{'}, x_{ij_4}^{'})$. Next, we observe that as we shift from the allocation $(x_{ij_3}^{*}, x_{ij_4}^{*})$ to $(x_{ij_3}^{'}, x_{ij_4}^{'})$, we "save" money given by $\epsilon (p_{j_4} - p_{j_3})$, but also reduce utility by $\epsilon (u_{ij_4} - u_{ij_3})$. But, now, we can use this "saved" money to purchase $\theta_a$.

So consider, $(x_{ij_1}^{'}, x_{ij_2}^{'})$, such that $x_{ij_2}^{'} = x_{ij_2}^* - \tau$ and $x_{ij_1}^{'} = x_{ij_1}^* + \tau$, where $\tau$ satisfies $\tau (p_{j_2} - p_{j_1}) = \epsilon (p_{j_4} - p_{j_3})$, i.e., we use our unspent or "saved" money. As a result we gain a utility given by $\tau (u_{ij_2} - u_{ij_1})$. Then we have that the change in the utility of the agents is given by:

\[u_i^{'} - u_i^* = \tau (u_{ij_2} - u_{ij_1}) - \epsilon (u_{ij_4} - u_{ij_3}) = \epsilon (u_{ij_4} - u_{ij_3}) \cdot \left( \frac{\tau}{\epsilon} \cdot \frac{u_{ij_2} - u_{ij_1}}{u_{ij_4} - u_{ij_3}} - 1 \right) \]
Here, we have taken $u_i^{'}$ as the utility corresponding to the allocation $(x_{ij_1}^{'}, x_{ij_2}^{'}, x_{ij_3}^{'}, x_{ij_4}^{'})$ and $u_i^*$ as the utility corresponding to the allocation $(x_{ij_1}^{*}, x_{ij_2}^{*}, x_{ij_3}^{*}, x_{ij_4}^{*})$. Next, since we have that $\tau (p_{j_2} - p_{j_1}) = \epsilon (p_{j_4} - p_{j_3})$, it follows that:

\[u_i^{'} - u_i^* = \epsilon (u_{ij_4} - u_{ij_3}) \cdot \left( \frac{p_{j_4} - p_{j_3}}{p_{j_2} - p_{j_1}} \cdot \frac{u_{ij_2} - u_{ij_1}}{u_{ij_4} - u_{ij_3}} - 1 \right) = \epsilon (u_{ij_4} - u_{ij_3}) \cdot \left( \frac{\theta_b}{\theta_a} - 1 \right) >0 \]

This establishes that $(x_{ij_1}^{*}, x_{ij_2}^{*}, x_{ij_3}^{*}, x_{ij_4}^{*})$ cannot be optimal and so we have a contradiction. Thus, our claim follows.

\subsection{Proof of Corollary~\ref{cor:coroptsol}} \label{cor-optimal-pf}

We now restate and proceed to prove Corollary~\ref{cor:coroptsol}.

\coroptsol*

\begin{proof}
An agent $i$ purchases \textit{virtual products} until their budget is completely used up or all the \textit{virtual products} are bought, which implies that it is just the last \textit{virtual product} that the agent purchases, which may be bought for less than one unit. As the remaining \textit{virtual products} are all purchased completely, these correspond to endpoints of the line segments of the lower frontier of the convex hull of the solution set $S_t$. This indicates that agent $i$ either does not buy any good (if we are at the origin in Figure~\ref{virtual-products}) or buys one unit of one good only corresponding to an end point of one of the line segments characterizing the lower frontier of the convex hull. Since the last \textit{virtual product} that the agent purchases may be bought for less than one unit, this corresponds any point on the lower frontier of the convex hull. In particular, if such a point lies in the middle of a line segment on the lower frontier of the convex hull then we have that the agent buys a corresponding fraction of two goods (associated with the line segment) within that resource \textit{type}.
\end{proof}

\subsection{Remark on Optimal Solution of \textbf{IOP}} \label{rmk-iopopt}

The optimal solution of the \textbf{IOP} can also be characterized for goods that do not have any physical constraints. For such a product $j$ without physical constraints, we can define the value $\theta_j = \frac{p_j}{u_{ij}}$ and augment this to the list of the $\theta$ values for each of the \textit{virtual products}. As we ordered the \textit{virtual products} in the descending order of their \textit{bang-per-buck} $ = \frac{1}{\theta_{j_1j_2}}$ for \textit{virtual products} corresponding to goods $j_1$ and $j_2$, we can do the same for this new list of $\theta$ values that includes $\theta_j$.

Now buyer $i$ will purchase goods in the ascending order of the $\theta$ values that includes $\theta_j$ analogous to the result in Theorem~\ref{thm:iopoptimal}. As in Corollary~\ref{cor:coroptsol}, at most one unit of each \textit{virtual product} can be purchased; however, any amount of good $j$ (the good without physical constraints) can be purchased. Thus, if agent $i$ still has budget at $\theta_j$, then the remaining budget will be used to purchase product $j$.

\subsection{Examples of Optimal Solution of \textbf{IOP}} \label{examples-iopopt}

We provide two examples to illustrate the use of \textit{virtual products} in characterizing the optimal solution of the \textbf{IOP}. For both examples, we consider the market with $6$ products and an agent $i$. Furthermore, we denote the utilities of the agent for each of the products as: $u_{i1} = 1, u_{i2} = 2, u_{i3} = 3, u_{i4} = 4, u_{i5} = 5, u_{i6} = 6$ and the price that the agent observes in the market as: $p_1 = 0.1, p_2 = 0.4, p_3 = 0.7, p_4 = 1.2, p_5 = 1.7, p_6 = 2.4$. The difference between the two examples we consider for this market will be in terms of the physical constraints and budgets. In particular, the first example we consider will be one wherein each of the goods in the market has physical constraints and the second example will be one in which one of the goods, i.e., good $5$, will not have any physical constraint.

\textbf{Example 1} (All goods have Physical constraints): We endow agent $i$ with a budget of $w_i = 2.4$, and consider the following physical constraint for \textit{type} 1 goods: $x_{i1} + x_{i3} + x_{i5} \leq 1$, and for \textit{type} 2 goods: $x_{i2} + x_{i4} + x_{i6} \leq 1$.

Given the price and the utilities, we can derive the $\theta$ values for the \textit{virtual products} by calculating the slopes of the lower frontier of the convex hull in Figure~\ref{example1-iop-opt}, giving $\theta_1 = 0.1, \theta_2 = 0.2, \theta_3 = 0.3, \theta_4 = 0.4, \theta_5 = 0.5, \theta_6 = 0.6$.

\begin{figure}[!h]
      \centering
      \includegraphics[width=0.9\linewidth]{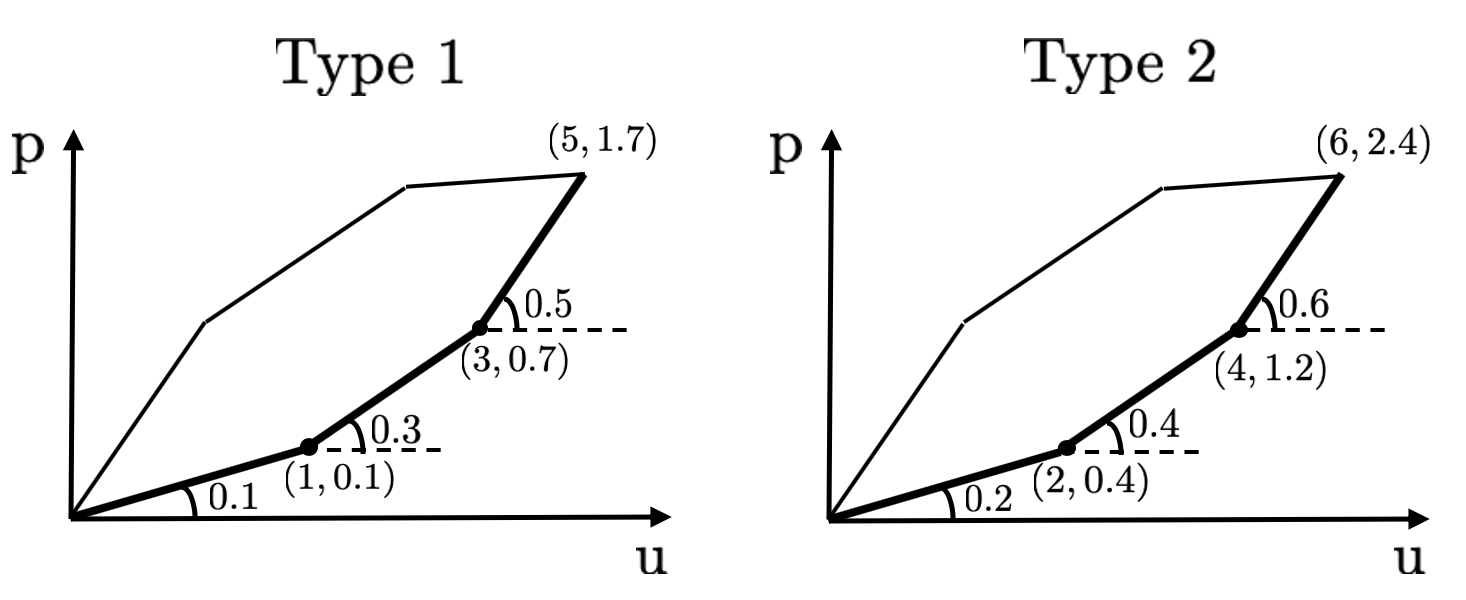}
      \caption{The solution set of the goods in example 1 for both the resource \textit{types} are depicted along with their corresponding $\theta$ values, i.e., the slope of the line segments on the lower frontier of the convex hull (shown in bold).}
      \label{example1-iop-opt}
   \end{figure}

We also note that for a \textit{virtual product} corresponding to points $A = (u_{ij_1}, p_{j_1})$ and $B = (u_{ij_2}, p_{j_2})$, the price is $p_{j_2} - p_{j_1}$ and the utility is $u_{ij_2} - u_{ij_1}$, with slope $\theta_{j_1j_2} = \frac{p_{j_2} - p_{j_1}}{u_{ij_2} - u_{ij_1}}$.

Now, as the agent will buy products in the ascending order of the $\theta$ values until the agent's budget is exhausted by Theorem~\ref{thm:iopoptimal}, we will have that the buyer purchases the following \textit{virtual products}:

\begin{itemize}
    \item 1 unit of $\theta_1$, with a total cost of $0.1$
    \item 1 unit of $\theta_2$, with a total cost of $0.4$
    \item 1 unit of $\theta_3$, with a total cost of $0.6$
    \item 1 unit of $\theta_4$, with a total cost of $0.8$
    \item 0.5 units of $\theta_5$, with a total cost of $0.5$
\end{itemize}

Mapping the purchase of these products to Figure~\ref{example1-iop-opt}, we observe that since both $\theta_2$ and $\theta_4$ are bought 1 unit, we have that this corresponds to the point on the lower frontier of the convex hull corresponding to good 4 and so $x_4^* = 1$. Furthermore, since $0.5$ units of $\theta_5$ are purchased, this maps to a point on the line segment corresponding to good 3 and good 5. Thus, we have that $x_3^* = 0.5$ and $x_5^* = 0.5$. Thus, the final allocations for the agent are: $x_1^* = 0, x_2^* = 0, x_3^* = 0.5, x_4^* = 1, x_5^* = 0.5, x_6^* = 0$.

We also note that these allocations confirm the claim of Corollary~\ref{cor:coroptsol}.

\textbf{Example 2} (Good 5 has no Physical constraints): We endow agent $i$ with a budget of $w_i = 4.5$, and consider the following physical constraint for \textit{type} 1 goods: $x_{i1} + x_{i3} \leq 1$, and for \textit{type} 2 goods: $x_{i2} + x_{i4} + x_{i6} \leq 1$. Good $5$ does not have any physical constraints.

The calculations of the $\theta$ values of all products other than $\theta_5$ remain the same. We calculate $\theta_5$ as mentioned in the remark in Appendix section~\ref{rmk-iopopt}, and so $\theta_5 = \frac{1.7}{5} = 0.34$.

Now, as the agent will buy products in the ascending order of the $\theta$ values until the agent's budget is exhausted by Theorem~\ref{thm:iopoptimal}, we will have that the buyer purchases the following \textit{virtual products}:

\begin{itemize}
    \item 1 unit of $\theta_1$, with a total cost of $0.1$
    \item 1 unit of $\theta_2$, with a total cost of $0.4$
    \item 1 unit of $\theta_3$, with a total cost of $0.6$
    \item 2 units of $\theta_5$, with a total cost of $3.4$
\end{itemize}

The buyer purchases $\theta_5$ before $\theta_4$, as $\theta_5 = 0.34 < 0.4 = \theta_4$. Further, since good 5 does not have any physical constraint, as mentioned in Appendix section~\ref{rmk-iopopt}, the buyer can purchase an arbitrary amount of the good 5 until their budget is exhausted.

We observe that since both $\theta_1$ and $\theta_3$ are bought 1 unit, we have that this corresponds to the point on the lower frontier of the convex hull (for \textit{type} 1 goods) corresponding to good 3 and so $x_3^* = 1$. Next, since one unit of $\theta_2$ is bought and no units of $\theta_4$ are bought, we have that $x_2^* = 1$. Finally, since $2$ units of $\theta_5$ are purchased, as this good does not correspond to any physical constraints, we have $x_5^* = 2$.

Thus, the final allocations for the agent are: $x_1^* = 0, x_2^* = 1, x_3^* = 1, x_4^* = 0, x_5^* = 2, x_6^* = 0$. We note that this allocation is in line with the remark in Appendix section~\ref{rmk-iopopt}.

\subsection{Proof of Theorem~\ref{thm:imposs}} \label{Appendix1}

We now restate and proceed to prove Theorem~\ref{thm:imposs}.

\imposs*

To prove Theorem~\ref{thm:imposs} we need to compare the KKT conditions of the individual and social optimization problems. As in the proof sketch, we only derive the KKT conditions of the social optimization problem \textbf{SOP1} and show that the budgets of agents will in general not be completely used up and thus a market clearing equilibrium cannot hold.

\subsubsection{KKT conditions of Social Optimization Problem \textbf{SOP1}}

We now derive the first order necessary and sufficient KKT conditions for \textbf{SOP1}. To do so, we introduce the dual variables $\mathbf{p} \in \mathbb{R}_{\geq 0}^{m}$ for constraint~\eqref{eq:ImpossSocOpt1}, $\mathbf{r}_{i} \in \mathbb{R}_{\geq 0}^{l_i}, \forall i$ for each of the constraints in~\eqref{eq:ImpossSocOpt2} and $s_{ij} \leq 0, \forall i, j$ for each of the non-negativity constraints~\eqref{eq:ImpossSocOpt3}. Next, we derive the Lagrangian of this problem as:

\begin{align}
    \mathcal{L} = \sum_{i = 1}^{n} w_i \log \left(\sum_{j = 1}^{m} u_{ij} x_{ij} \right) - \sum_{j = 1}^{m} p_j \left(\sum_i x_{ij} - \Bar{s}_j \right) - \sum_{i = 1}^{n} \sum_{t = 1}^{l_i} r_{it} \left(\sum_{j = 1}^{m} A_{tj}^{(i)} x_{ij} - 1 \right) - \sum_{i = 1}^{n} \sum_{j = 1}^{m} s_{ij} x_{ij} 
\end{align}
The first order derivative condition is found by taking the derivative of the Lagrangian with respect to $x_{ij}$:

\begin{align}\label{eq:ImpossLagDerSoc}
    \frac{w_i}{\sum_{j = 1}^{m} u_{ij} x_{ij}} u_{ij} - p_j - \sum_{t = 1}^{l_i} r_{it} A_{tj}^{(i)} \leq 0 
\end{align}
Next, the complimentary slackness condition for this problem can be derived by multiplying~\eqref{eq:ImpossLagDerSoc} by $x_{ij}$:

\begin{align}\label{eq:ImpossCompSlackSoc1}
    \frac{w_i}{\sum_{j = 1}^{m} u_{ij} x_{ij}} u_{ij} x_{ij} - p_j x_{ij} - \sum_{t = 1}^{l_i} r_{it} A_{tj}^{(i)} x_{ij} = 0 
\end{align}
Thus, the KKT conditions of the social optimization problem are given by equations~\eqref{eq:ImpossSocOpt1}-\eqref{eq:ImpossSocOpt3},~\eqref{eq:ImpossLagDerSoc},~\eqref{eq:ImpossCompSlackSoc1} and the sign constraints of the dual variables.

\subsubsection{Establishing Theorem~\ref{thm:imposs}} \label{AppendixImposs}

We now use the above derivations of the first order necessary and sufficient KKT conditions to complete the proof of Theorem~\ref{thm:imposs}.

\begin{proof}
We observe that at the market clearing outcome, it must be that the budgets of all agents must be entirely used up. We now show that this will in general not be true. To see this, consider the KKT condition in equation~\eqref{eq:ImpossCompSlackSoc1} of the social optimization problem. If we sum over $j$, we observe that this equation can be expressed as:

\begin{align}\label{eq:ImpossCompSlackSoc2}
    w_i - \sum_{j = 1}^{m} p_j x_{ij} - \sum_{t = 1}^{l_i} r_{it} = 0 
\end{align}
However, the above constraint implies that the only way $w_i = \sum_{j = 1}^{m} p_j x_{ij}$ is if $\sum_{t = 1}^{l_i} r_{it} = 0$, which implies that $r_{it} = 0$ $\forall i, t$, as each $r_{it} \geq 0$. However, this in general cannot be expected, as at the market clearing outcome the physical constraint~\eqref{eq:ImpossSocOpt2} may be met with equality for some or all agents (as was observed in the examples in Appendix sections~\ref{non-existence-prop1},~\ref{non-uniqueness-prop2} and~\ref{examples-iopopt}). In particular, $r_{it} = 0$ $\forall i, t$ implies that if we loosen any of the physical constraints then the objective function value will remain unchanged. However, if there is a difference between agent's utilities between goods then for the \textit{type} of goods for which the physical constraint is loosened at least one person will be better off as they have higher utilities and so the objective function must increase. As a result, in general the scenario $r_{it} = 0$ $\forall i, t$ is not possible, particularly in cases when physical constraints are met with equality, which may be the case for many markets. Thus, we have that the market clearing KKT conditions of the social and individual optimization problems are not necessarily the same, establishing our claim.

\end{proof}

\subsection{Proof of Theorem~\ref{thm:thm2}} \label{proof-thm2}

We now restate and proceed to prove Theorem~\ref{thm:thm2}.

\thmm*

We now derive the KKT conditions for the budget perturbed social optimization problem \textbf{BP-SOP} and individual optimization problem \textbf{IOP} that are used in the derivation of Theorem~\ref{thm:thm2}.

\subsubsection{Derivation of KKT Conditions for Social Optimization Problem \textbf{BP-SOP}} \label{Appendix4}

To establish the equivalence between individual optimization problem~\eqref{eq:eq1}-\eqref{eq:con3} and the perturbed social optimization problem \textbf{BP-SOP}, we start by deriving the first order necessary and sufficient KKT conditions for the problem~\eqref{eq:SocOpt}-\eqref{eq:SocOpt3}. To do so, we introduce the dual variables $\mathbf{p} \in \mathbb{R}_{\geq 0}^{m}$ for constraint~\eqref{eq:SocOpt1}, $\mathbf{r}_{i} \in \mathbb{R}_{\geq 0}^{l_i}, \forall i$ for each of the constraints in~\eqref{eq:SocOpt2} and $s_{ij} \leq 0, \forall i, j$ for each of the non-negativity constraints~\eqref{eq:SocOpt3}. Next, we derive the Lagrangian of this problem as:

\begin{align}
    \mathcal{L} = \sum_{i = 1}^{n} (w_i + \lambda_i) \log \left(\sum_{j = 1}^{m} u_{ij} x_{ij} \right) - \sum_{j = 1}^{m} p_j \left(\sum_i x_{ij} - \Bar{s}_j \right) - \sum_{i = 1}^{n} \sum_{t = 1}^{l_i} r_{it} \left(\sum_{j = 1}^{m} A_{tj}^{(i)} x_{ij} - 1 \right) - \sum_{i = 1}^{n} \sum_{j = 1}^{m} s_{ij} x_{ij} 
\end{align}
The first order derivative condition is found by taking the derivative of the Lagrangian with respect to $x_{ij}$:

\begin{align}\label{eq:LagDerSoc}
    \frac{w_i + \lambda_i}{\sum_{j = 1}^{m} u_{ij} x_{ij}} u_{ij} - p_j - \sum_{t = 1}^{l_i} r_{it} A_{tj}^{(i)} \leq 0 
\end{align}
Next, the complimentary slackness condition for this problem can be derived by multiplying~\eqref{eq:LagDerSoc} by $x_{ij}$:

\begin{align}\label{eq:CompSlackSoc1}
    \frac{w_i + \lambda_i}{\sum_{j = 1}^{m} u_{ij} x_{ij}} u_{ij} x_{ij} - p_j x_{ij} - \sum_{t = 1}^{l_i} r_{it} A_{tj}^{(i)} x_{ij} = 0 
\end{align}
Thus, the KKT conditions of the social optimization problem \textbf{BP-SOP} are given by equations~\eqref{eq:SocOpt1}-\eqref{eq:SocOpt3},~\eqref{eq:LagDerSoc},~\eqref{eq:CompSlackSoc1} and the sign constraints of the dual variables.

\subsubsection{KKT conditions of Individual Optimization Problem \textbf{IOP}} \label{Appendix4_2}

We now derive the KKT conditions of \textbf{IOP} by formulating a Lagrangian and introducing the dual variable $y_i \geq 0$ for~\eqref{eq:con1}, $\mathbf{\Tilde{r}}_{i} \in \mathbb{R}_{\geq 0}^{l_i}$ for each of the constraints in~\eqref{eq:con2} and $\Tilde{s}_{ij} \leq 0, \forall j$ for each of the non-negativity constraints~\eqref{eq:con3}. Thus, our Lagrangian is:

\begin{align}
    \mathcal{L}(\mathbf{x}_i, y_i, \mathbf{r}_i, \mathbf{s}) = \sum_{j = 1}^{m} u_{ij} x_{ij} - y_i \left (\sum_{j = 1}^{m} p_j x_{ij} - w_i \right) - \sum_{t = 1}^{l_i} \Tilde{r}_{it} \left (\sum_{j = 1}^{m} A_{tj}^{(i)} x_{ij} - 1 \right) - \sum_{j = 1}^{m} \Tilde{s}_{ij} x_{ij} 
\end{align}
The first order constraint for this problem is found by taking the derivative of our Lagrangian with respect to $x_{ij}$ and noting that $\Tilde{s}_{ij} \leq 0$:

\begin{align}\label{eq:LagDer}
    u_{ij} - y_i p_j - \sum_{t = 1}^{l_i} \Tilde{r}_{it} A_{tj}^{(i)} \leq 0 
\end{align}
Next, we derive the complimentary slackness condition for this problem by multiplying~\eqref{eq:LagDer} by $x_{ij}$:

\begin{align}\label{eq:CompSlack1}
    u_{ij}x_{ij} - y_i p_j x_{ij} - \sum_{t = 1}^{l_i} \Tilde{r}_{it} A_{tj}^{(i)} x_{ij} = 0 
\end{align}
Thus, the final KKT conditions for our problem are given by equations~\eqref{eq:con1},~\eqref{eq:con2},~\eqref{eq:con3},~\eqref{eq:LagDer} and~\eqref{eq:CompSlack1}, and the sign constraints on the dual variables. Furthermore, at equilibrium conditions, it must hold that the goods must be sold to capacity, which is given by: $\sum_{i = 1}^{n} x_{ij} = \Bar{s}_j, \forall j$.

We use the derived KKT conditions of the individual and social optimization problems to prove our theorem in the analysis that follows.

\begin{proof}
($\Rightarrow$) To prove the forward direction of our claim, we need to show that given a market equilibrium $(\mathbf{x}_i, \mathbf{p})$, $\forall i$ of the \textbf{IOP}, we can construct $\lambda_i$, such that $\lambda_i = \sum_{t = 1}^{l_i} r_{it}$, $\forall i$.

We proceed by considering two cases, (i) when $y_i>0$ and (ii) when $y_i = 0$, where $y_i$ is the dual variable of the budget constraint in the \textbf{IOP}. We first note at equilibrium conditions $\sum_{i = 1}^{n} x_{ij} = \Bar{s}_j, \forall j \in [m]$ and that the constraints of the individual optimization problem \textbf{IOP} already implies the other constraints of \textbf{BP-SOP}, i.e., the constraints~\eqref{eq:SocOpt2} and~\eqref{eq:SocOpt3}. Thus, all we need to do is check the lagrangian derivative condition~\eqref{eq:LagDerSoc}, complementary slackness condition~\eqref{eq:CompSlackSoc1} and dual multiplier sign constraints.

\textbf{Case 1} ($y_i>0$): When we sum over all goods $j$, we have from the complimentary slackness condition~\eqref{eq:CompSlack1} that:

\begin{align} \label{eq:comp_slack_iopsop}
    \sum_{j = 1}^{m} u_{ij}x_{ij} = y_i \sum_{j = 1}^{m} p_j x_{ij} + \sum_{j = 1}^{m} \sum_{t = 1}^{l_i} \Tilde{r}_{it} A_{tj}^{(i)} x_{ij} = y_i w_i + \sum_{t = 1}^{l_i} \Tilde{r}_{it}
\end{align}
The second equality follows from the complimentary slackness condition for the budget and physical constraints.

Next, using the Lagrangian derivative condition~\eqref{eq:LagDer}, and taking $r_{it} = \frac{\Tilde{r}_{it}}{y_i}$ we have that:

\begin{align} \label{eq:new-ineq-iopsop}
    \frac{1}{y_i} \leq \frac{p_j + \sum_{t = 1}^{l_i} r_{it} A_{tj}^{(i)}} {u_{ij}}, \forall j
\end{align}

Now setting $\lambda_i = \sum_{t = 1}^{l_i} r_{it} = \frac{\sum_{t = 1}^{l_i} \Tilde{r}_{it}}{y_i}$, we obtain that:

\begin{align}
    \frac{w_i + \lambda_i}{\sum_{j = 1}^{m} u_{ij}x_{ij}} = \frac{1}{y_i} \cdot \frac{ \left(y_i w_i + \sum_{t = 1}^{l_i} \Tilde{r}_{it} \right)}{\sum_{j = 1}^{m} u_{ij}x_{ij}} 
\end{align}
Now observing Equation~\eqref{eq:comp_slack_iopsop} and using the inequality~\eqref{eq:new-ineq-iopsop}, we can rewrite the above expression as:

\begin{align} \label{eq:reprove-comp-slack}
    \frac{w_i + \lambda_i}{\sum_{j = 1}^{m} u_{ij}x_{ij}} = \frac{1}{y_i} \leq \frac{p_j + \sum_{t = 1}^{l_i} r_{it} A_{tj}^{(i)}} {u_{ij}}, \forall j
\end{align}
This is exactly the first order derivative condition of the \textbf{BP-SOP} as in Equation~\eqref{eq:LagDerSoc}. Multiplying Equation~\eqref{eq:reprove-comp-slack} by $x_{ij}$, we obtain by complimentary slackness the condition in equation~\eqref{eq:CompSlackSoc1}.

Finally, since we have that all the dual variables in \textbf{IOP} are scaled by a positive constant, as $y_i>0$ the corresponding signs of the dual variables in $\textbf{BP-SOP}$ remain intact. Thus, we have shown that when $y_i>0$ and when we set $\lambda_i = \sum_{t = 1}^{l_i} r_{it}$, $\forall i$ then the market equilibrium KKT conditions of the \textbf{IOP} are the same as that of $\textbf{BP-SOP}$.

\textbf{Case 2} ($y_i = 0$): We start by introducing some necessary notation. From Theorem~\ref{thm:iopoptimal}, we define $\alpha_i^*$ to be the slope of buyer $i$'s last bought \textit{virtual product}, i.e., with the maximal slope, and let $\lambda_i = \alpha_i^*(\sum_{j = 1}^{m} u_{ij} x_{ij}) - w_i$.

Finally, for good $j$ belonging to \textit{type} $t$, we define $\sum_{t' = 1}^{l_i} r_{it'} A_{t'j}^{(i)} = r_{it} = \alpha_i^* u_{ij} - p_j$ for $x_{ij}>0$, which follows as $A_{tj}^{(i)} = 1$ for a unique \textit{type} $t$ for each good $j$. Then, we obtain that:

\begin{align}
    \frac{w_i + \lambda_i}{\sum_{j = 1}^{m} u_{ij}x_{ij}} = \alpha_i^* = \frac{p_j + \sum_{t' = 1}^{l_i} r_{it'} A_{t'j}^{(i)}}{u_{ij}}, \forall j \text{ s.t. } x_{ij}>0
\end{align}

On the other hand if $x_{ij} = 0$, we have that $\frac{p_{j_1} - p_j}{u_{ij_1} - u_{ij}} = \theta_t \leq \alpha_i^*$, which implies that: $\alpha_i^* u_{ij} - p_j \leq \alpha_i^* u_{ij_1} - p_{j_1} = r_{it} = \sum_{t' = 1}^{l_i} r_{it'} A_{t'j}^{(i)}$, where $j, j_1 \in t$.

Using this relationship, we observe that:

\begin{align}
    \frac{w_i + \lambda_i}{\sum_{j = 1}^{m} u_{ij}x_{ij}} = \alpha_i^* \leq \frac{p_j + \sum_{t' = 1}^{l_i} r_{it'} A_{t'j}^{(i)}}{u_{ij}}, \forall j \text{ s.t. } x_{ij}=0
\end{align}
This is exactly the first order derivative condition of the \textbf{BP-SOP} as in Equation~\eqref{eq:LagDerSoc}. Multiplying Equation~\eqref{eq:reprove-comp-slack} by $x_{ij}$, we obtain by complimentary slackness the condition in equation~\eqref{eq:CompSlackSoc1}.

We now note that to show $\lambda_i = \sum_{t = 1}^{l_i} r_{it}$, we multiply the equation $\sum_{t' = 1}^{l_i} r_{it'} A_{t'j}^{(i)} = r_{it} = \alpha_i^* u_{ij} - p_j$ with $x_{ij}$ and sum over all goods $j$ to get:

\begin{align}
    \sum_{j = 1}^{m} \sum_{t' = 1}^{l_i} r_{it'} A_{t'j}^{(i)} x_{ij} = \sum_{j = 1}^{m} \alpha_i^* u_{ij} x_{ij} - \sum_{j = 1}^{m} p_j x_{ij}
\end{align}
Then by complimentary slackness we get that:

\begin{align}
    \sum_{t' = 1}^{l_i} r_{it'} = \sum_{j = 1}^{m} \alpha_i^* u_{ij} x_{ij} - \sum_{j = 1}^{m} p_j x_{ij} = \alpha_i^* \sum_{j = 1}^{m} u_{ij} x_{ij} - w_i = \lambda_i
\end{align}

Next, we observe that since $\alpha_i^*$ is the maximal $\theta$, we must have that $\alpha_i^* \geq \frac{p_j}{u_{ij}}$ which implies that $\alpha_i^* u_{ij} - p_j \geq 0$ giving us that $r_{it}\geq 0$.

Finally, both $r_{it} \geq 0$, $\forall i, t$ and $p_j \geq 0$, $\forall j$ and so the signs of the dual variables in $\textbf{BP-SOP}$ remain intact. Thus, we have shown that when $y_i=0$ and when we set $\lambda_i = \sum_{t = 1}^{l_i} r_{it}$, $\forall i$ then the market equilibrium KKT conditions of the \textbf{IOP} are the same as that of $\textbf{BP-SOP}$.

($\Leftarrow$) To establish the converse of the theorem, we first note that the constraints of the social optimization problem already implies that $\forall j$, good $j$ is sold to capacity if $p_j>0$. If on the other hand $p_{j_1} = 0$ for some good $j_1 \in t$, i.e., good $j_1$ belongs to good type $t$ and $\sum_i x_{ij_1} < \Bar{s}_{j_1}$ then it must follow that $\forall i$, $\sum_{j \in t} x_{ij} = 1$, as agents can always buy more units of good type $j_1$ to increase their utility. Without loss of generality, suppose that all agents are interested in consuming goods from resource type $t$, so that $t \in T_i$, $\forall i$. Thus, we have that: $n = \sum_{i = 1}^{n} \sum_{j \in t} x_{ij} = \sum_{j \in t} \sum_{i = 1}^{n} x_{ij} < \sum_{j \in t} \Bar{s}_{j} $, as $\sum_i x_{ij_1} < \Bar{s}_{j_1}$ by assumption. However, for \textbf{BP-SOP} to have a feasible solution it must be that $\sum_{j \in t} \Bar{s}_{j} \leq n$. Combining this with the above inequality we get that $n<n$, a contradiction. Thus, we must have that $\sum_i x_{ij_1} = \Bar{s}_{j_1}$. As a result, we have established that $\forall j$, good $j$ is sold to capacity.

The constraints of the social optimization problem also imply the \textbf{IOP} constraints~\eqref{eq:con2} and~\eqref{eq:con3}. Next, we use the notation and equations in appendix section~\ref{Appendix4} and the corresponding KKT conditions of \textbf{BP-SOP} and \textbf{IOP}.
\newline
If we divide~\eqref{eq:LagDerSoc} and~\eqref{eq:CompSlackSoc1} by $\frac{w_i + \lambda_i}{\sum_{j = 1}^{m} u_{ij} x_{ij}}$, then these equations map respectively to the Lagrangian derivative equation~\eqref{eq:LagDer} and complimentary slackness equation~\eqref{eq:CompSlack1} of the individual optimization problem. We show this by denoting $\mathbf{p} \in \mathbb{R}_{\geq 0}^{m}$ as the dual variable for the capacity constraint, i.e., $p_j$ is the dual variable corresponding to the capacity constraint for good $j$, and then we have that~\eqref{eq:LagDerSoc} becomes:

\begin{align}\label{eq:LagDerSoc2}
    u_{ij} - \frac{\sum_{j = 1}^{m} u_{ij} x_{ij}}{w_i + \lambda_i} p_j - \sum_{t = 1}^{l_1} \frac{\sum_{j = 1}^{m} u_{ij} x_{ij}}{w_i + \lambda_i}r_{it} A_{tj}^{(i)} \leq 0 
\end{align}
The above equation is equivalent to ~\eqref{eq:LagDer}, where $y_i = \frac{\sum_{j = 1}^{m} u_{ij} x_{ij}}{w_i + \lambda_i} \geq 0$, $\Tilde{r}_{it} = \frac{\sum_{j = 1}^{m} u_{ij} x_{ij}}{w_i + \lambda_i}r_{it} \geq 0$. The same analysis holds for the complementary slackness condition equivalence in the two problems. 

Now, all it remains for us to satisfy is the \textbf{IOP} constraint~\eqref{eq:con1}. To do this, we use the complementary slackness condition in equation~\eqref{eq:CompSlackSoc1} and sum over $j$. Then realizing that $r_{it} \sum_{j = 1}^{m} A_{tj}^{(i)} x_{ij} = r_{it}$ by complimentary slackness we get:

\begin{align}\label{eq:CompSlackSoc2}
    w_i + \lambda_i - \sum_{j = 1}^{m} p_j x_{ij} - \sum_{t = 1}^{l_i} r_{it} = 0 
\end{align}
Thus, if we set $\lambda_i = \sum_{t = 1}^{l_i} r_{it}$ as in the statement of the theorem, we obtain that the budget constraint condition~\eqref{eq:con1} is satisfied with equality. This implies that we have a market clearing outcome, as the budgets are completely used and the goods are sold to capacity, while the KKT conditions of the two problems are equivalent at market clearing conditions. This establishes our claim.

\end{proof}

\clearpage

\end{document}